\documentclass[11pt,a4paper]{article}

\usepackage{preamble-fv-simple}
\usepackage{macros}

\author[1]{Lars Jaffke} 
\author[1]{Paloma T.\ Lima} 
\author[2]{Daniel Lokshtanov} 

\affil[1]{University of Bergen, Norway} 
\affil[ ]{\texttt{\{lars.jaffke,paloma.lima\}@uib.no}} 

\affil[2]{UC Santa Barbara, California, USA} 
\affil[ ]{\texttt{daniello@ucsb.edu}} 

\title{$b$-Coloring Parameterized by Clique-Width\thanks{%
	Based on an extended abstract that appeared STACS 2021~\cite{stacs2021}.
}} 

\begin{document}

\maketitle

\begin{abstract}	
	We provide a polynomial-time algorithm for \textsc{$b$-Coloring} on graphs of constant clique-width.
	This unifies and extends nearly all previously known polynomial time results on graph classes,
	and answers open questions posed by Campos and Silva [Algorithmica, 2018]
	and Bonomo et al.~[Graphs Combin., 2009].
	This constitutes the first result concerning structural parameterizations of this problem.
	We show that the problem is \FPT when parameterized by the vertex cover number on general graphs,
	and on chordal graphs when parameterized by the number of colors.
	Additionally, we observe that our algorithm for graphs of bounded clique-width
	can be adapted to solve the \textsc{Fall Coloring} problem within the same runtime bound.
	The running times of the clique-width based algorithms for \bcol and \fallcol 
	are tight under the Exponential Time Hypothesis.
\end{abstract}


\section{Introduction}
This paper settles open questions regarding the complexity 
of the \textsc{$b$-Coloring} problem on graph classes
and initiates the study of its structural parameterizations.
A \emph{$b$-coloring} of a graph $G$ with $k$ colors is a partition of the vertices of $G$ 
into $k$ independent sets such that
each of them contains a vertex that has a neighbor in all of the remaining ones.
The \emph{$b$-chromatic number} of $G$, denoted by $\chi_b(G)$, is the maximum integer $k$
such that $G$ admits a $b$-coloring with $k$ colors.
This notion was introduced by Irving and Manlove~\cite{IrvingManlove1999} to describe the
behavior of the following color-suppressing heuristic for the \textsc{Graph Coloring} problem.
We start with some proper coloring of the input graph $G$ 
and try to iteratively suppress one of its colors.
That is, for a given color $c$, we consider each vertex $v$ of color $c$, and check if there is another color $c' \neq c$
available that does not appear in its neighborhood.
If so, we assign vertex $v$ the color $c'$, observing that the coloring remains proper, 
and repeat this process for the remaining vertices of color $c$.
If successful, we remove the color $c$ from all vertices of $G$ and decrease the number of colors by one.
Once no color can be supressed by this procedure, the coloring at hand is a $b$-coloring of $G$, and in the worst case,
this heuristic produces a coloring with $\chi_b(G)$ many colors.

Since then, the \bcol and \bchrom problems which given a graph $G$ and an integer $k$ ask whether $G$ has a $b$-coloring with $k$ colors
and whether $\chi_b(G) \ge k$, respectively, have received considerable attention in the algorithms and complexity communities.
Before we discuss these results, 
note that the \bcol and \bchrom problem are not as closely related as the \graphcol and \textsc{Chromatic Number} problems 
in terms of their (polynomial time) complexities.
If we can solve \textsc{Chromatic Number}, then we can use this algorithm to solve \graphcol,
since each $n$-vertex graph $G$ has proper colorings with $\chi(G), \ldots, n$ colors.
However, knowing $\chi_b(G)$ and $\chi(G)$ does not say anything about the existence of a $b$-coloring with $k \in \{\chi(G) + 1, \ldots, \chi_b(G) - 1\}$ colors.
Therefore, the \bcol problem can be computationally harder on a graph class than the \bchrom problem.
Trivially, if we know how to solve \bcol in polynomial time, we can solve \bchrom in polynomial time.

The \bchrom problem has been shown to be \NP-complete in the general case~\cite{IrvingManlove1999},
as well as on bipartite graphs~\cite{KratochvilTuzaVoigt2002},
co-bipartite graphs~\cite{BonomoEtAl2015},
chordal graphs~\cite{HavetSalesSampaio2012},
and line graphs~\cite{CamposEtAl2015}, and
a lot of effort has been put into devising polynomial time algorithms
for \bcol in various other classes of graphs.%
\footnote{In many of the following references, the results are stated for \bchrom instead of \bcol;
however the algorithms for \bcol follow from the algorithms for \bchrom together with 
the fact that these graph classes are \emph{$b$-continuous}%
~\cite{BonomoEtAl2015,Faik2004,VBK2011},
meaning that they have $b$-colorings any number $k \in \{\chi(G), \ldots, \chi_b(G)\}$ of colors,
and the fact that \textsc{Chromatic Number} is solvable in polynomial time on these graph classes 
(for instance via~\cite{EspelageEtAl2001,Wan94}).
}
These include trees~\cite{IrvingManlove1999},
tree-cographs~\cite{BonomoEtAl2015},
and graphs with few $P_4$s, such as cographs and $P_4$-sparse graphs~\cite{BonomoEtAl2009},
$P_4$-tidy graphs~\cite{VBK2011}, and $(q, q-4)$-graphs for constant~$q$~\cite{CamposEtAl2014}.
A common property shared by these graph classes is that they all have bounded 
\emph{clique-width}~\cite{GolumbicRotics2000,Gurski2017,MakowskyRotics1999,Vanherpe2004}.\footnote{To the best of our knowledge,
the only polynomial time result for graphs of unbounded clique-width so far concerns graphs of large girth.
In particular, Campos et al.~\cite{CamposEtAl2015b} showed that 
\bchrom is polyomial-time solvable on graphs of girth at least $7$.}

The main contribution of this work is an algorithm that solves \textsc{$b$-Coloring} (and \textsc{$b$-Chromatic Number}) 
in polynomial time on graphs of constant clique-width.
Besides unifying the above mentioned polynomial time cases, 
this extends the tractability landscape of these problems to larger graph classes,
and answers two open problems stated in the literature.

Over a decade ago, Bonomo et al.~\cite{BonomoEtAl2009} asked whether their polynomial time result for cographs
can be extended to distance-hereditary graphs.
Havet et al.~\cite{HavetSalesSampaio2012} answered the question negatively by providing an 
\NP-completeness proof for chordal distance-hereditary graphs.
We observe, however, that their proof has a flaw 
and while it does prove the claimed statement for chordal graphs, 
it unfortunately fails to do so for distance-hereditary graphs.
Our polynomial time algorithm for graphs of bounded clique-width in fact provides a positive answer to Bonomo et al.'s question,
as distance-hereditary graphs have clique-width at most three~\cite{GolumbicRotics2000}.
In recent years, even subclasses of distance-hereditary graphs have received significant attention,
for instance in the work of Campos and Silva~\cite{CamposSilva2018}:
they provide a polynomial time
algorithm for claw-free block graphs, and ask whether this result can be generalized to block graphs.
Our algorithm provides a positive answer to this question as well.
Moreover, it extends the known algorithm for $(q, q-4)$-graphs~\cite{CamposEtAl2014} (for constant~$q$) 
to all $(q, t)$-graphs for constants $q$ and $t$ with $q \ge 4$, $t \ge 0$, 
and either $q \le 6$ and $t \le q - 4$, or $q \ge 7$ and $t \le q - 3$,
by a theorem due to Makowsky and Rotics~\cite{MakowskyRotics1999}.
Similarly, it extends the polynomial time algorithm for $P_4$-tidy graphs~\cite{VBK2011}
to the class of partner-limited graphs thanks to a result by Vanherpe~\cite{Vanherpe2004}.
We give an overview of the graph classes involved in the previous discussion in Figure~\ref{fig:graphclasses}.
\begin{figure}
	\centering
	\includegraphics[width=.75\textwidth]{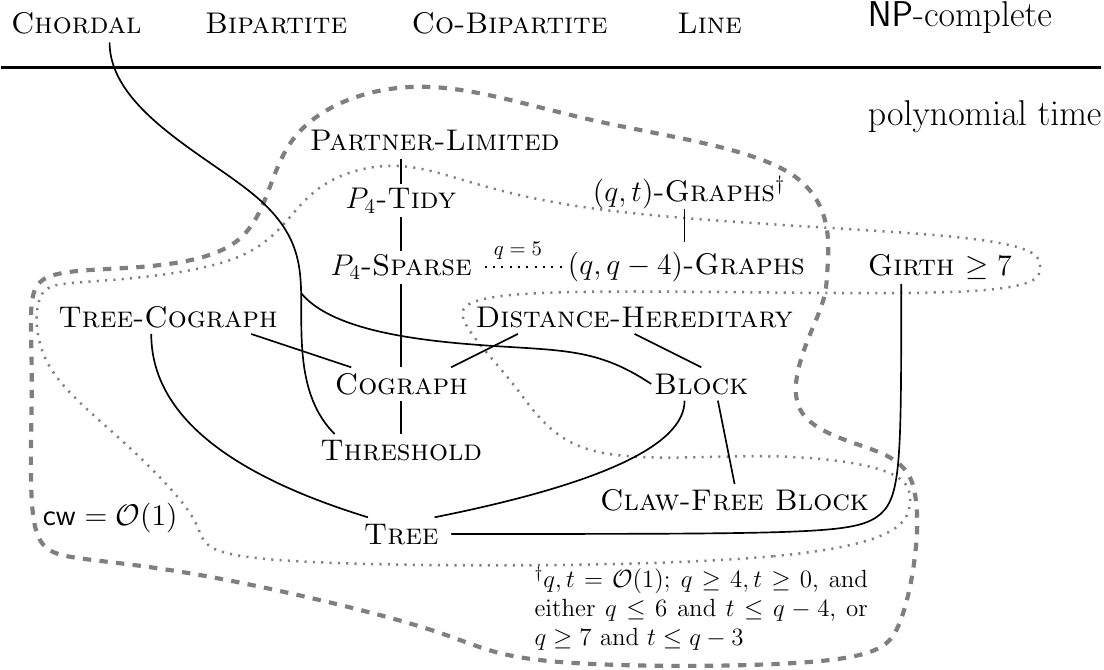}
	\caption{Some graph classes on which the complexities of \bcol and \bchrom problem were studied. 
		Whenever two classes are connected by a line, the upper one contains the lower one. 
		All \NP-hardness results hold for \bchrom 
		and all polynomial time results, except the one for graphs of girth at least seven,
		hold for \bcol.
		The inner bottom area (dotted line) shows classes for which polynomial time algorithms were previously known 
		and the outer area (dashed line, labeled $\cw = \calO(1)$) shows on which classes our algorithm can be applied.}
	\label{fig:graphclasses}
\end{figure}

Our algorithm runs in time $n^{2^{\calO(\givenmw)}}$, 
where $n$ denotes the number of vertices of the input graph 
which is given together with a clique-width $\givenmw$-expression.
As consequences of results due to Fomin et al.~\cite{FominEtAl2010} and Fomin et al.~\cite{Fomin2018},
we observe that \bcol parameterized by clique-width is \W[1]-hard, and that
the exponential dependence on $\givenmw$
in the degree of the polynomial cannot be avoided unless the Exponential Time Hypothesis (\ETH) fails.
Concretely, an algorithm running in time $n^{2^{o(\givenmw)}}$ would refute \ETH.

From the point of view of parameterized complexity,
Panolan et al.~\cite{PanolanPhilipSaurabh2017} showed that \bchrom
parameterized by the number of colors is \W[1]-hard.
However, this problem may even be harder,
since so far no \XP-algorithm is known.
Recently, Aboulker et al.~\cite{AboulkerEtAl2020} showed that the more restrictive 
\textsc{$b$-Chromatic Core} problem parameterized by the number of colors
(which has a brute-force \XP-algorithm, see e.g.~\cite{EffantinEtAl2016}) remains \W[1]-hard.

It is therefore natural to ask which additional restrictions can be imposed to obtain 
parameterized tractability results.
For instance, an open problem posed by Sampaio~\cite{SampaioThesis} (see also~\cite{SilvaThesis}) 
asks whether \textsc{$b$-Coloring} parameterized by the number of colors is \FPT on chordal graphs.
We answer this question in the affirmative.
%
Other restricted cases that have been considered in the literature  
target specific numbers of colors that depend on the input graph.
The \textsc{Dual $b$-Coloring} problem, which asks if an input $n$-vertex graph has a $b$-coloring with $n-k$ colors,
is \FPT parameterized by $k$~\cite{HavetSampaio2013}.
Moreover, deciding if a graph $G$ has a $b$-coloring with $k = \Delta(G) + 1$ colors, which is an upper bound on $\chi_b(G)$, 
is \FPT parameterized by $k$~\cite{PanolanPhilipSaurabh2017,SampaioThesis},
while the case $k = \Delta(G)$ is \XP and for every fixed $p \ge 1$,
the case $k = \Delta(G) - p$ is \NP-complete for $k = 3$~\cite{JaffkeLima2020}.

Another novelty aspect of our \XP-algorithm parameterized by clique-width 
is that it is the first result about 
\emph{structural parameterizations} of the \bcol and \bchrom problems.
In all previously known polynomial time cases the algorithms only work if the input graph has 
some prescribed structure. Our algorithm works on all graphs,
albeit with a prohibitively slow runtime on graphs of large clique-width.
In this vein, we round off our work with an \FPT-result for another lead player among structural parameterizations,
the \emph{vertex cover number} of a graph; a parameter often referred to as the \emph{Drosophila}
of parameterized complexity.

\paragraph*{Fall Coloring.}~ 
A \emph{fall coloring} is a special type of $b$-coloring where 
\emph{every} vertex needs to have at least one neighbor in all color classes except its own.
In other words, it is a partition of the vertex set of a graph into independent dominating sets.
As a standalone notion, fall coloring has been introduced by Dunbar et al.~\cite{DunbarEtAl2000}.
However, since the corresponding \fallcol problem
falls in the category of locally checkable vertex partitioning problems,
it has been shown in earlier work of Telle and Proskurowski~\cite{TelleProskurowski1997} 
to be \FPT parameterized by the tree-width of the input graph,
as well as \FPT parameterized by clique-width plus the number of colors by Gerber and Kobler~\cite{GerberKobler2003} 
(see also~\cite{Bui-XuanEtAl2013}),
and by Heggernes and Telle~\cite{HeggernesTelle1998} 
to be \NP-complete for fixed number of colors. 
\fallcol remains hard further restricted to
bipartite~\cite{LaskarLyle2009,LauriMitillos2019,Silva2019},
chordal~\cite{Silva2019}, or
planar~\cite{LauriMitillos2019} graphs.
On the other hand, even with unbounded number of colors, it
is known to be solvable in polynomial time on 
strongly chordal graphs~\cite{LyleEtAl2005,GoddardHenning2013}, 
threshold graphs and split graphs~\cite{Mitillos2016}.
In all of these cases, one simply checks whether the chromatic number of the input graph
is equal to its minimum degree plus one.
To the best of our knowledge, these are the only known polynomial time cases.

We adapt our algorithm for \bcol on graphs of bounded clique-width to solve 
\fallcol, and therefore show that the latter problem is as well solvable in time 
$n^{2^{\calO(\givenmw)}}$, where $\givenmw$ denotes the clique-width of a given 
decomposition of the input graph.
By a simple reduction, we show that \fallcol is also \W[1]-hard in this parameterization 
and that an $n^{2^{o(\givenmw)}}$-time algorithm for it would refute \ETH.

\paragraph*{Vertex Coloring Problems Parameterized by Clique-Width.}~
We briefly touch on differences in the complexities of vertex coloring problems of graphs
when parameterized by clique-width.
While the standard \gcol problem, asking for a proper coloring of the input graph,
is \XP-time solvable parameterized by clique-width~\cite{EspelageEtAl2001,Wan94},
some of its generalizations are \NP-complete on graphs of constant clique-width.
In the \textsc{List Coloring} problem we are given a graph $G$ and for each of its vertices $v$
a list $L(v)$ of colors, and the question is whether $G$ has a proper coloring such that each vertex
is assigned a color from its list.
This problem is \NP-complete on the (not disjoint) union of two complete graphs~\cite{Jansen1997}.
We can see that such graphs have bounded clique-width for instance by observing that they do not contain a path on four vertices as an induced subgraph, and are therefore cographs, which have clique-width at most two~\cite{CourcelleOlariu2000}.
%
In the related \textsc{Precoloring Extension} problem, we are given a graph, 
some of whose vertices already received a color, and the question is 
whether this coloring can be extended to a proper coloring of the entire graph.
The following standard reduction from \textsc{List Coloring}, 
starting with a graph that is the union of two complete graphs, 
shows that this variant is \NP-complete on graphs of constant clique-width as well.
Take the graph $G$ together with the lists $L(\cdot)$, and construct a graph $H$ by adding to $G$, for each vertex $v \in V(G)$
and each color $c \notin L(v)$, a new vertex $x_v^c$ which is adjacent only to $v$ and assigned color $c$.
It is not difficult to see that this precoloring of $H$ can be extended to the remainder of its vertices if and only if 
$G$ has a list coloring using the lists $L(\cdot)$.
Moreover, adding pendant vertices to a graph does not increase its clique-width.

Belmonte et al.~\cite{BelmonteEtAl2020} showed that the \textsc{Grundy Coloring} problem,
which asks for a linear order of the vertices that maximizes the number of colors used by the greedy coloring heuristic,
is \NP-complete on graphs of constant clique-width.
This nicely contrasts our \XP-algorithm for \bcol,
since both the \bcol and the \textsc{Grundy Coloring} problems are rooted in the theoretical analysis
of graph coloring heuristics.

Very recently, Jaffke et al.~\cite{JaffkeLimaPhilip2020} showed that the \textsc{Clique Coloring} problem,
asking for a vertex coloring without monochromatic maximal cliques, 
is \XP parameterized by clique-width as well. 
The question whether \textsc{Clique Coloring} parameterized by clique-width is \W[1]-hard remains open.

\paragraph*{Sketch of the algorithm.}~
Let us discuss how we obtain our \XP-algorithm parameterized by clique-width.
First, we consider a branch decomposition of the input graph $G$ of bounded \emph{module-width} $\givenmw$
which is equivalent to clique-width and has the following property.
At each node $t$ of the branch decomposition we have a subgraph $G_t$ of $G$
whose vertex set can be partitioned into at most $\givenmw$ equivalence classes 
with respect to their neighborhood outside of $G_t$.
For the purpose of our dynamic programming algorithm, it suffices to describe colorings
by the way each of their color classes interacts with these equivalence classes.
In the \gcol problem, it is enough to describe a color class according to its intersection
with the equivalence classes of $G_t$ alone~\cite{EspelageEtAl2001,Wan94} (see also~\cite{Fomin2018}).
For the \bcol problem, we additionally have to ensure that eventually, each color class
indeed has a $b$-vertex. 
The challenge is to do so without explicitly remembering which color classes a vertex 
has already seen in its neighborhood -- this would result in prohibitively large tables.
We overcome this difficulty by a symmetry breaking trick that instead stores, 
for each color class,
a \emph{demand} to the future neighbors of the equivalence classes which 
-- if fulfilled -- guarantees that the \emph{other} color classes 
can have $b$-vertices in the end.


\section{Preliminaries}
\paragraph*{Graphs.}~
All graphs considered here are simple and finite. 
For a graph $G$ we denote by $V(G)$ and $E(G)$ the vertex set and edge set of $G$, respectively.
For an edge $e = uv \in E(G)$, we call $u$ and $v$ the \emph{endpoints} of $e$ and we write $u \in e$ and $v \in e$. 

For two graphs $G$ and $H$, we say that $G$ is a \emph{subgraph} of $H$, written $G \subseteq H$,
if $V(G) \subseteq V(H)$ and $E(G) \subseteq E(H)$.
For a set of vertices $S \subseteq V(G)$, the \emph{subgraph of $G$ induced by $S$}
is $G[S] \defeq (S, \{uv \in E(G) \mid u, v \in S\})$.

For a graph $G$ and a vertex $v \in V(G)$, the set of its \emph{neighbors} is 
$N_G(v) \defeq \{u \in V(G) \mid uv \in E(G)\}$, and the \emph{degree} of $v$ is $\deg_G(v) \defeq \card{N_G(v)}$.
The \emph{closed neighborhood} of $v$ is $N_G[v] \defeq \{v\} \cup N_G(v)$.
For a set $X \subseteq V(G)$, we let $N_G(X) \defeq \bigcup_{v \in X} N_G(v) \setminus X$
and $N_G[X] \defeq X \cup N_G(X)$.
In all these cases, we may drop $G$ as a subscript if it is clear from the context.
A graph is called \emph{subcubic} if all its vertices have degree at most three.

A graph $G$ is \emph{connected} if for all $2$-partitions $(X, Y)$ of $V(G)$ 
with $X \neq \emptyset$ and $Y \neq \emptyset$,
there is a pair $x \in X$, $y \in Y$ such that $xy \in E(G)$.
A \emph{connected component} of a graph is a maximal connected subgraph.
A connected graph is called a \emph{cycle} if all its vertices have degree two.
A connected graph is called a \emph{tree} if it has no cycle as a subgraph.
In a tree $T$, the vertices of degree one are called the \emph{leaves} of $T$,
denoted by $\leaves(T)$, and the vertices in $V(T) \setminus \leaves(T)$ are 
the \emph{internal vertices} of $T$.
A tree of maximum degree at most two is a \emph{path} and the leaves of 
a path are called its \emph{endpoints}.
If $P$ is a path with endpoints $u$ and $v$, then we say that $P$ is a 
\emph{path from $u$ to $v$}.
The \emph{length} of a path is the number of its edges.
For a graph $G$ and a pair of vertices $u, v \in V(G)$,
we denote by $\dist_G(u, v)$ the length of the shortest path between $u$ and $v$ in $G$.

A tree $T$ is called \emph{rooted}, if there is a distinguished vertex $r \in V(T)$,
called the \emph{root} of $T$, inducing an ancestral relation on $V(T)$:
for a vertex $v \in V(T)$, if $v \neq r$, the neighbor of $v$ on the path 
from $v$ to $r$ is called the \emph{parent} of $v$, and all other neighbors of $v$ 
are called its \emph{children}.
For a vertex $v \in V(T) \setminus \{r\}$ with parent $p$, 
the \emph{subtree rooted at $v$}, denoted by $T_v$, 
is the subgraph of $T$ induced by all vertices 
that are in the same connected component of 
$(V(T), E(T) \setminus \{vp\})$ as $v$. 
We define $T_r \defeq T$.
A tree $T$ is called a \emph{caterpillar} if it contains a path $P \subseteq T$ such that
all vertices in $V(T) \setminus V(P)$ are adjacent to a vertex in $P$.

For a graph $H$, we say that a graph $G$ is \emph{$H$-free} 
if $G$ does not contain $H$ as an induced subgraph.
For a set of graphs $\calH$, we say that $G$ is \emph{$\calH$-free}
if $G$ is $H$-free for all $H \in \calH$.
For an integer $k \ge 3$, let $C_k$ denote a cycle on $k$ vertices.
A graph $G$ is called \emph{chordal} if it is $\{C_n \mid n \ge 4\}$-free.
A graph $G$ is called \emph{distance-hereditary} 
if for each connected induced subgraph $H$ of $G$, and each pair of vertices $u, v \in V(H)$,
$\dist_H(u, v) = \dist_G(u, v)$.

A set of vertices $S \subseteq V(G)$ of a graph $G$ is called an \emph{independent set}
if $E(G[S]) = \emptyset$.
A set of vertices $S \subseteq V(G)$ is a \emph{vertex cover} in $G$ if $V(G) \setminus S$ 
is an independent set in $G$.
A set of vertices $S \subseteq V(G)$ is a \emph{clique} in $G$ if $E(G[S]) = \{uv \mid u, v \in S\}$.

A graph $G$ is called \emph{bipartite} if its vertex set can be partitioned into two
nonempty independent sets, which we will refer to as a \emph{bipartition} of $G$.

\paragraph*{Notation for Equivalence Relations.}~
Let $\Omega$ be a set and $\sim$ an equivalence relation over $\Omega$.
For an element $x \in \Omega$ the \emph{equivalence class of $x$}, denoted by $[x]_\sim$
or simply $[x]$ if $\sim$ is clear from the context,
is the set $\{y \in \Omega \mid x \sim y\}$.
We denote the set of all equivalence classes of $\sim$ by $\Omega/{\sim}$.

\paragraph*{Parameterized Complexity.}~ 
We give the basic definitions of parameterized complexity that are relevant to this work 
and refer to~\cite{CyganEtAl2015,DowneyFellows2013} for details.
Let $\Sigma$ be an alphabet. A \emph{parameterized problem} is a set $\Pi \subseteq \Sigma^* \times \bN$,
the second component being the \emph{parameter} which usually expresses a structural measure of the input.
A parameterized problem $\Pi$ is said to be \emph{fixed-parameter tractable}, or in the complexity class \FPT,
if there is an algorithm that for any $(x, k) \in \Sigma^* \times \bN$ 
correctly decides whether or not $(x, k) \in \Pi$, and runs in time $f(k) \cdot \card{x}^c$
for some computable function $f \colon \bN \to \bN$ and constant $c$.
We say that a parameterized problem is in the complexity class \XP, 
if there is an algorithm that for each $(x, k) \in \Sigma^* \times \bN$ correctly decides whether or not $(x, k) \in \Pi$,
and runs in time $f(k) \cdot \card{x}^{g(k)}$, for some computable functions $f$ and $g$.

The concept analogous to \NP-hardness in parameterized complexity is that of \W[1]-hardness, 
whose formal definition we omit.
The basic assumption is that $\FPT \neq \W[1]$, under which no \W[1]-hard problem admits an \FPT-algorithm.
For more details, see~\cite{CyganEtAl2015,DowneyFellows2013}.

\paragraph*{Exponential Time Hypothesis.}~
The \textsc{$3$-SAT} problem asks whether 
a given boolean formula in conjunctive normal form 
with clauses of size at most three
has a truth assignment to its variables that lets the formula evaluate to true.
In 2001, Impagliazzo, Paturi, and Zane~\cite{ImpagliazzoPaturi2001,ImpagliazzoPaturiZane2001}
conjectured that any algorithm for the \textsc{$3$-SAT} problem
requires exponential time.
This conjecture is known as the \emph{Exponential Time Hypothesis} (\ETH)
whose plausibility stems from the fact that despite numerous efforts,
a subexponential-time algorithm for \textsc{$3$-SAT} remains elusive. 
It can be stated as follows.
\begin{conjecture-nn}[\ETH~\cite{ImpagliazzoPaturi2001,ImpagliazzoPaturiZane2001}]
	There is no algorithm that solves each instance of \textsc{$3$-SAT} on $n$ variables
	in time $2^{o(n)}$.
\end{conjecture-nn}

This conjecture initiated a rich theory of hardness results conditioned on \ETH 
(see e.g.\ the survey~\cite{LokshtanovMarxSaurabh2011} and~\cite[Chapter 14]{CyganEtAl2015}),
allowing for more precise lower bounds than the ones obtained from assumptions 
such as $\P \neq \NP$ or $\FPT \neq \W[1]$.

\subsection{Clique-Width, Branch Decompositions, and Module-Width}
We first define clique-width, 
introduced by Courcelle, Engelfriet, and Rozenberg~\cite{CourcelleEtAl1993},
and then the equivalent measure of \emph{module-width} that we will use in our algorithm.
We keep the definition of clique-width slightly informal 
and refer to~\cite{CourcelleEtAl1993,CourcelleOlariu2000} for more details.
The reason why we choose module-width over clique-width is that module-width allows 
for a slightly more compact description of our algorithm, 
since it suffices to consider a single operation in the dynamic programming 
instead of several.

Let $G$ be a graph. The \emph{clique-width} of $G$, denoted by $\cliquewidth(G)$,
is the minimum number of labels $\{1, \ldots, k\}$ 
needed to obtain $G$ using the following four operations:
\begin{enumerate}
	\item Create a new graph consisting of a single vertex labeled $i$.
	\item Take the disjoint union of two labeled graphs $G_1$ and $G_2$.
	\item Add all edges between pairs of vertices of label $i$ and label $j$.
	\item Relabel every vertex labeled $i$ to label $j$.
\end{enumerate}

We now turn to the definition of module-width which is based on 
the notion of a rooted branch decomposition.
\begin{definition}[Branch decomposition]
	Let $G$ be a graph.
	A \emph{branch decomposition} of $G$ is a pair $(T, \decf)$ of a subcubic tree $T$ and a bijection $\decf \colon V(G) \to \leaves(T)$.
	If $T$ is a caterpillar, then $(T, \decf)$ is called \emph{linear branch decomposition}.
	If $T$ is rooted, then we call $(T, \decf)$ a \emph{rooted branch decomposition}.
	In this case, for $t \in V(T)$, we denote by $T_t$ the subtree of $T$ rooted at $t$, 
	and we define $V_t \defeq \{v \in V(G) \mid \decf(v) \in \leaves(T_t)\}$,
	$\overline{V_t} \defeq V(G) \setminus V_t$, and
	$G_t \defeq G[V_t]$.
\end{definition}

Module-width is attributed to Rao~\cite{Rao2006,Rao2008}.\footnote{Note that in \cite{Rao2008}, 
module-width is referred to as \emph{modular-width} which usually has a different meaning, see e.g.~\cite{Coudert2019}.}
On a high level, the module-width of a rooted branch decomposition measures, at each of its nodes $t$,
the number of subsets of $\overline{V_t}$ that make up the intersection of $\overline{V_t}$ 
with the neighborhood of some vertex in $V_t$.
This naturally groups the vertices of $V_t$ into equivalence classes.
\begin{definition}[Module-width]
	Let $G$ be a graph, and $(T, \decf)$ be a rooted branch decomposition of $G$. 
	For each $t \in V(T)$, let $\sim_t$ be the equivalence relation on $V_t$ defined as follows:
	\begin{align*}
		\forall u, v \in V_t \colon u \sim_t v \Leftrightarrow N_G(u) \cap \overline{V_t} = N_G(v) \cap \overline{V_t}
	\end{align*}
	
	The \emph{module-width} of $(T, \decf)$ is $\mw(T, \decf) \defeq \max_{t \in V(T)} \card{V_t/{\sim_t}}$.
	The \emph{module-width of $G$}, denoted by $\modulew(G)$, 
	is the minimum module width over all rooted branch decompositions of $G$.
\end{definition}

\begin{theorem}[Rao, Thm.~6.6 in~\cite{Rao2006}]\label{thm:cw:mw}
	For any graph $G$, $\mw(G) \le \cw(G) \le 2 \cdot \mw(G)$,
	and given a decomposition of bounded clique-width, a decomposition of bounded module-width,
	and vice versa,
	can be constructed in time $\calO(n^2)$, where $n = \card{V(G)}$.
\end{theorem}

\paragraph*{The operator $(\decaux_t, \bubblemap_r, \bubblemap_s)$ of node $t$ with children $r$ and $s$.}
Let $(T, \decf)$ be a rooted branch decomposition of a graph $G$ and let $t \in V(T)$ be a node with children $r$ and $s$.
We now describe an operator associated with $t$ that tells us how the graph $G_t$ is formed from its subgraphs $G_r$ and $G_s$,
and how the equivalence classes of $\sim_t$ are formed from the equivalence classes of $\sim_r$ and $\sim_s$.
First, it is clear that $V_t = V_r \cup V_s$.
Since $G_r$ and $G_s$ are induced subgraphs of $G_t$, 
we furthermore know that $E(G_t[V_r]) = E(G_r)$ and $E(G_t[V_s]) = E(G_s)$,
so it remains to describe the edges between $V_r$ and $V_s$.
By the definition of module-width, we know that each pair of vertices $u, v \in V_r$ with $u \sim_r v$
has the same neighborhood in $\overline{V_r} = V_s \cup \overline{V_t}$.
Hence, for each vertex $z \in V_s$, 
we know that \emph{either both or neither} of $u$ and $v$ are adjacent to $z$.
In other words, for each pair $Q_r \in V_r/{\sim_r}$, $Q_s \in V_s{\sim_s}$, 
either all edges between each pair of a vertex from $Q_r$ and a vertex from $Q_s$ are present in $G_t$, or none of them.
This can be described by a bipartite graph $\decaux_t$ on bipartition $(V_r/{\sim_r}, V_s/{\sim_s})$ 
with $[u]_{\sim_r}[v]_{\sim_s} \in E(\decaux_t)$ if and only if $uv \in E(G_t)$. To summarize,
\begin{align*}
	E(G_t)\,&= E(G_r) \cup E(G_s) \cup F \mbox{ where } \\
		F\,&= \{uv \mid u \in V_r, v \in V_s, \{[u]_{\sim_r}, [v]_{\sim_s}\} \in E(\decaux_t)\}.
\end{align*}

By roughly the same reasoning, we can observe that the equivalence relation $\sim_t$ 
\emph{coarsens} the equivalence relations $\sim_r$ and $\sim_s$.
Consider again vertices $u, v \in V_r$ such that $u \sim_r v$.
Then, $N(u) \cap \overline{V_r} = N(v) \cap \overline{V_r}$,
and since $V_r \subseteq V_t$ we have that $\overline{V_t} \subseteq \overline{V_r}$,
which implies that $N(u) \cap \overline{V_t} = N(v) \cap \overline{V_t}$, so $u \sim_t v$.
However, it may well be that there are vertices $u, v \in V_r$ with $u \not\sim_r v$,
but $u \sim_t v$: this is the case when $u$ and $v$ have the same neighbors in $\overline{V_t}$,
but different neighbors in $V_s$.
Lastly, note that there may also be vertices $v \in V_r$ and $z \in V_s$ such that $v \sim_t z$.

We have argued that each equivalence class of $\sim_t$ can be obtained by taking a subset of equivalence classes of 
$\sim_r$ and $\sim_s$, and joining them (in what we call a `bubble' below).
Formally,
there is a partition $\calP = \{P_1, \ldots, P_h\}$ of $V(\decaux_t) = V_r/{\sim_r} \cup V_s/{\sim_s}$ 
such that $V_t/{\sim_t} = \{Q_1, \ldots, Q_h\}$,
where for $1 \le i \le h$, $Q_i = \bigcup_{Q \in P_i} Q$.
For each $1 \le i \le h$, we call $P_i$ the \emph{bubble} of the resulting equivalence class $\bigcup_{Q \in P_i} Q$ of $\sim_t$.

As auxiliary structures, for $p \in \{r, s\}$, we let $\bubblemap_p \colon V_p/{\sim_p} \to V_t/{\sim_t}$ 
be the map such that for all $Q_p \in V_p/{\sim_p}$, $Q_p \subseteq \bubblemap_p(Q_p)$,
i.e.\ $\bubblemap_p(Q_p)$ is the equivalence class of $\sim_t$ whose bubble contains $Q_p$.
We call $(\decaux_t, \bubblemap_r, \bubblemap_s)$ the \emph{operator} of $t$.

\subsection{Colorings}
Let $G$ be a graph. 
An ordered partition $\calC = (C_1, \ldots, C_k)$ of $V(G)$ is called a \emph{coloring} of $G$ (with $k$ colors).
(Observe that for $i \in \{1, \ldots, k\}$, $C_i$ may be empty.)
For $i \in \{1, \ldots, k\}$, we call $C_i$ the \emph{color class $i$},
and say that the vertices in $C_i$ \emph{have color $i$}.
$\calC$ is called \emph{proper} if for all $i \in \{1, \ldots, k\}$, $C_i$ is an independent set in $G$.
The \emph{restriction} of a coloring $\calC = (C_1, \ldots, C_k)$ to a vertex set $S \subseteq V(G)$, 
is $\calC|_{S} \defeq (C_1 \cap S, \ldots, C_k \cap S)$.
In this case we say conversely that $\calC$ \emph{extends} $\calC|_S$.

Whenever convenient, we may alternatively denote a coloring of a graph with $k$ colors as a map
$\phi\colon V(G) \to \{1, \ldots, k\}$.
In this case, a restriction of $\phi$ to $S$ is the map $\phi|_S\colon S \to \{1, \ldots, k\}$
with $\phi|_S(v) = \phi(v)$ for all $v \in S$.
For any $T \subseteq V(G)$ with $S \subseteq T$, we say that $\phi|_T$ extends $\phi|_S$.

A proper coloring $(C_1, \ldots, C_k)$ is called a \emph{$b$-coloring}, 
if for all $i \in \{1, \ldots, k\}$, there is a vertex $v_i \in C_i$,
called \emph{$b$-vertex of color $i$}, 
such that for all $j \in \{1, \ldots, k\} \setminus \{i\}$, $N_G(v_i) \cap C_j \neq \emptyset$.
In this work, we study the following computational problem.
\fancyproblemdef
	{$b$-Coloring}
	{Graph $G$, integer $k$}
	{Does $G$ have a $b$-coloring with $k$ colors?}

We sometimes denote a $b$-coloring $\calC = (C_1, \ldots, C_k)$ by $(\calC, B = \{v_1, \ldots, v_k\})$, 
where for all $i \in \{1, \ldots, k\}$, $v_i$ is a $b$-vertex of color $i$.
In this case, $B$ can be understood as the set containing a \emph{witness $b$-vertex} for each color class. 

The following definition will be key to the algorithms presented in the next sections.
\begin{definition}[Partial $b$-Coloring]\label{def:partbcol}
	Let $G$ be a graph and $k \in \bN$.
	For an induced subgraph $H$ of $G$,
	a \emph{partial $b$-coloring} of $H$ is a pair $(\calC, B)$ of a proper coloring $\calC = (C_1, \ldots, C_k)$ of $H$
	and a subset $B \subseteq V(H)$ such that for all $i \in [k]$, $\card{C_i \cap B} \le 1$.
	We call the vertices in $B$ the \emph{partial $b$-vertices}.
\end{definition}

\subsection{Distance-hereditary graphs}\label{sec:bcol:distance:hereditary}
In their work on $P_4$-sparse graphs, Bonomo et al.~\cite{BonomoEtAl2009} asked whether $b$-{\sc Coloring} is polynomial-time solvable on the class of distance-hereditary graphs. Havet et al.~\cite{HavetSalesSampaio2012} claimed to answer this question in the negative way, showing that $b$-{\sc Coloring} is \NP-complete on chordal distance-hereditary graphs. Their proof, however, contains a flaw and the graph constructed in their reduction, even though indeed chordal, fails to be distance-hereditary. In what follows, we briefly describe their reduction and argue that the graph constructed is not distance-hereditary.

The reduction presented in~\cite{HavetSalesSampaio2012} is from {\sc 3-Edge Coloring} restricted to the class of 3-regular graphs. Given an instance $G$ for {\sc 3-Edge Coloring} with $V(G)=\{v_1,\ldots,v_n\}$, they construct a graph $H$ as follows. The vertex set of $H$ contains a copy of $V(G)$ plus one vertex associated with each edge of $G$. We denote by $e_{xy}$ the vertex corresponding to the edge $xy$. The vertices of $V(G)$ form a clique in $H$, the vertices corresponding to edges form an independent set, and for each edge $xy\in E(G)$, the vertex $e_{xy}$ is adjacent to the copy of $x$ and $y$ in $H$. The connected component of $H$ induced by these vertices is therefore a split graph. Finally, they add three disjoint copies of $K_{1,n+2}$ to $H$. It is thus easy to see that $H$ is a chordal graph. However, let $xz$ and $yz$ be two edges of $G$ sharing one endpoint. Then the subgraph of $H$ induced by $\{x,y,z,e_{xz},e_{yz}\}$ is isomorphic to a gem (see Figure~\ref{fig:flaw}). As shown by Bandelt and Mulder~\cite{Bandelt1986}, distance-hereditary graphs are gem-free graphs. This shows that the graph $H$ is not a distance-hereditary graph.
\begin{figure}
	\centering
	\includegraphics[height=1.8cm]{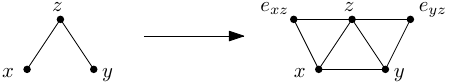}
	\caption{}
	\label{fig:flaw}
\end{figure}

\subsection{Parameterized by vertex cover}
In this subsection we prove that \bcol is \FPT when parameterized by vertex cover. We will do so by providing a $2^{\calO(\tw\cdot k)}n$ time algorithm for the problem parameterized by the tree-width of the input graph plus number of colors. The result for vertex cover will then follow from the fact that the vertex cover number of a graph is always at most its tree-width, and a $b$-coloring of a graph with vertex cover $\ell$ can have at most $\ell+1$ many colors. Indeed, either all $b$-vertices are contained in the vertex cover, in which case there are at most $\ell$ of them, or there is one outside, whose degree is at most $\ell$, and hence it can see at most $\ell$ colors in its neighborhood.
\begin{definition}
    Let $G$ be a graph.
    A \emph{tree decomposition} of $G$ is a pair $(T, \calB = \{B_t \mid t \in V(T)\})$, where $T$ is a tree, and the sets in $\calB$ are called \emph{bags},
    satisfying the following conditions.
    \begin{enumerate}
        \item $\bigcup_{t \in V(T)} B_t = V(G)$.
        \item For each $uv \in E(G)$, there is some $t \in V(T)$ such that $\{u, v\} \subseteq B_t$.
        \item For each $v \in V(G)$, $T[\{t \in V(T) \mid v \in B_t\}]$ is connected.
    \end{enumerate}
    The \emph{width} of a tree decomposition is $\max_{t \in V(T)} \card{B_t} - 1$ and the \emph{tree-width} of $G$ is the minimum width over all its tree decompositions.
\end{definition}
\begin{definition}
    A tree decomposition $(T, \calB = \{B_t \mid t \in V(T)\})$ of a graph $G$ is called \emph{nice}
    if $T$ is a rooted tree and each node $t \in V(T)$ is one of the following types:
    \begin{description}
        \item[Leaf:] $t$ is a leaf of $T$ and $B_t = \emptyset$.
        \item[Introduce:] $t$ has a single child $s$ and $B_t = B_s \cup \{v\}$ for some $v \in V(G)$; we say that $v$ is \emph{introduced} at $t$.
        \item[Forget:] $t$ has a single child $s$ and $B_s = B_t \cup \{v\}$ for some vertex $v \in B_t$; we say that $v$ is \emph{forgotten} at $t$.
        \item[Join:] $t$ has two children, $s_1$ and $s_2$, and $B_t = B_{s_1} = B_{s_2}$.
    \end{description}
    For $t \in V(T)$, we let $T_t$ denote the subtree of $T$ rooted at $t$; we let $V_t = \bigcup_{s \in V(T_t)} B_s$ and $G_t = G[V_t]$.
\end{definition}
\begin{theorem}[Korhonen~\cite{Korhonen2021}]\label{thm:Korhonen}
    There is an algorithm that given a graph $G$ on $n$ vertices and an integer $k$,
    in $2^{\calO(w)}n$ time
    either outputs a tree decomposition of $G$ of width at most $2k+1$ or concludes that the tree-width of $G$ is more than $k$.
\end{theorem}
\begin{lemma}[Kloks~\cite{Kloks1994}, verbatim from~\cite{CyganEtAl2015}]\label{lem:Kloks}
    If a graph $G$ admits a tree deecomposition of width at most $k$,
    then it also admits a nice tree decomposition of width at most $k$.
    Moreover, given a tree decomposition $(T, \calB)$ of $G$ of width at most $k$,
    one can in time $\calO(k^2\cdot\max\{\card{V(G)}, \card{V(T)}\})$
    find a nice tree decomposition of $G$ that has at most $\calO(k\card{V(G)})$ nodes.
\end{lemma}
\begin{proposition}\label{prop:tw:numcol}
    \textsc{$b$-Coloring} can be solved in $2^{\calO(\tw\cdot k)}n$ time,
    where $n$ is the number of vertices and $\tw$ the tree-width of the input graph, and $k$ the number of colors.
\end{proposition}
\begin{proof}
    By~\cref{thm:Korhonen,lem:Kloks} 
    we can assume that we have a nice tree decomposition $(T, \calB = \{B_t \mid t \in V(T)\})$ of $G$ of width $w \le 2\tw + 1$ and with $\calO(wn)$ nodes after spending $2^{\calO(\tw)}n$ time.
    We do bottom-up dynamic programming along $(T, \calB)$.
    
    The table entries of the dynamic programming and their invariant are as follows.
    Let $t \in V(T)$ be a node of $(T, \calB)$.
    Then, we let $\dptab_t[\gamma, C, P, \sigma] = 1$ if there is a partial $b$-coloring $\gamma_t$ of $G_t$ with the following properties, and $0$ otherwise:
    
    \begin{enumerate}
        \item $\gamma \colon B_t \to [k]$ is a proper coloring with $\gamma = \gamma_t|_{B_t}$.
        \item $P \subseteq B_t$ is the set of partial $b$-vertices of $\gamma_t$ that are contained in $B_t$.
        \item $\sigma \colon P \to 2^{[k]}$ is a map such that for each $p \in P$, $\sigma(p)$ is the set of colors that appear in the neighborhood of $p$ in $\gamma_t$.
        \item $C \subseteq [k]$, where $\gamma(P) \subseteq C$, is the set of colors that have a partial $b$-vertex in $\gamma_t$. Each partial $b$-vertex not contained in $B_t$ is a $b$-vertex.
    \end{enumerate}
    
    We observe that at each node $t \in V(T)$ there are at most $2^{\calO(wk)}$ table entries;
    moreover, once the table entries have been computed correctly,
    we know that $G$ has a $b$-coloring with $k$ colors if and only if at the root $\mathfrak{r}$ of $T$ there is a table entry $\dptab_{\mathfrak{r}}[\gamma, P, \sigma, C] = 1$, where $C = [k]$, and for all $p \in P$, $\sigma(p) = [k]$.
    We discuss how to compute the table entries for each type of node in $(T, \calB)$;
    we assume that initially all table entries are set to $0$.
    \begin{description}
        \item[Leaf.] If $t$ is a leaf, then it is trivial. For technical reasons, we assume that there is a table entry $\dptab_t[\emptyset, \emptyset, \emptyset, \emptyset] = 1$.
        \item[Introduce.] If $t$ is an introduce node, let $s$ be its child and $v$ the vertex introduced at $v$. 
        Let $\gamma$ be a proper $k$-coloring of $G[B_t]$. Each neighbor of $v$ that is a partial $b$-vertex for its color has to mark the color $\gamma(v)$ as seen in its neighborhood. To this end, for each $P \subseteq B_t$ and $\sigma\colon P \to 2^{[k]}$, we say that a map $\sigma_s\colon P \to 2^{[k]}$ is \emph{compatible with $\sigma$} 
        if for all $p \in P \cap N(v)$, $\sigma(p) = \sigma_s(p) \cup \{\gamma(v)\}$, 
        and for all $p \in P \setminus N[v]$, $\sigma(p) = \sigma_s(p)$.
        
        We first discuss how to deal with the case when $v$ is not a partial $b$-vertex for its color.
        We consider each set $C \subseteq [k]$, 
        each $P \subseteq B_t \setminus \{v\}$, 
        and each map $\sigma\colon P \to 2^{[k]}$. 
        We set $\dptab_t[\gamma, P, \sigma, C]$ to $1$ if there is a map $\sigma_s \colon P \to 2^{[k]}$ compatible with $\sigma$ and
        such that $\dptab_s[\gamma|_{B_s}, P, \sigma_s, C] = 1$.
        
        Next, we consider the case when $v$ is a partial $b$-vertex for its color.
        Then we consider each set $C \subseteq [k]$ with $\gamma(v) \in C$,
        and each $P \subseteq B_t$ with $v \in P$, 
        and each map $\sigma\colon P \to 2^{[k]}$ where $\sigma(v) = \gamma(N(v))$.
        We set $\dptab_t[\gamma, P, \sigma, C]$ to $1$ if there is a map $\sigma_s\colon P \setminus \{v\} \to 2^{[k]}$ that is compatible with $\sigma$ and such that $\dptab_s[\gamma|_{B_s}, P \setminus \{v\}, \sigma_s, C \setminus \{\gamma(v)\}] = 1$.
        \item[Forget.] If $t$ is a forget node, let $s$ be its child and $v$ be the vertex forgotten at $t$. The only thing we have to ensure here is that if $v$ was a partial $b$-vertex for its color, then in fact it was a $b$-vertex for its color.
        We proceed as follows.
        We set $\dptab_t[\gamma, P, \sigma, C]$ to $1$ if $\dptab_s[\gamma_s, P_s, \sigma_s, C] = 1$ where 
        $\gamma_s$ is an extension of $\gamma$ (assigning $v$ a color), 
        and either
        \begin{itemize}
            \item $v \notin P_s$, $P_s = P$, and $\sigma_s = \sigma$, or
            \item $v \in P_s$, $P = P_s \setminus \{v\}$, $\sigma_s|_{B_t} = \sigma$ and $\sigma_s(v) = [k]$.
        \end{itemize}
        \item[Join.] If $t$ is a join node, let $s_1$ and $s_2$ be its children.
        Here we only have to mark, for each partial $b$-vertex contained in $B_t$, the colors it has seen in $G_{s_1}$ and in $G_{s_2}$.
        Therefore we proceed as follows.
        We set $\dptab_t[\gamma, P, \sigma, C]$ to $1$ if 
        there exist $C_1, C_2 \subseteq [k]$ with $C_1 \cup C_2 = C$;
        and for $i \in [2]$, $\sigma_i \colon P \to 2^{[k]}$
        such that for all $p \in P$, $\sigma(p) = \sigma_1(p) \cup \sigma_2(p)$,
        and such that $\dptab_{s_i}[\gamma, P_i, \sigma_i, C_i] = 1$ for all $i \in [2]$.
    \end{description}
    
    Correctness of the algorithm follows from its description.
    Regarding its run time, we observe that for each node $t \in V(T)$, 
    all table entries $\dptab_t[\cdot]$ can be computed in time $2^{\calO(wk)}$. 
    Since the number of nodes in $T$ is at most $\calO(wn)$,
    the algorithm runs in time $2^{\calO(wk)}n = 2^{\calO(\tw\cdot k)}n$.
\end{proof}

\begin{corollary}\label{coro:vc}
    \textsc{$b$-Coloring} can be solved in $2^{\calO(\ell^2)}n$ time where
    $n$ is the number of vertices and $\ell$ the vertex cover number of the input graph.
\end{corollary}
\begin{proof}
    Let $G$ be the input graph with vertex cover number $\ell$.
    It is well-known that a vertex cover of size $\ell$ of $G$, which can be found in $\calO(1.2738^\ell + \ell n)$ time~\cite{ChenEtAl2010},
    can be used to give a path decomposition of $G$ of width (at most) $\ell$ in $\calO(n)$ time.
    Together with the fact that each $b$-coloring of a graph with vertex cover number $\ell$ can have at most $\ell + 1$ colors, the result follows from \cref{prop:tw:numcol}.
\end{proof}

\subsection{Chordal graphs}
Another consequence of \cref{prop:tw:numcol} is that \textsc{$b$-Coloring} is fixed-parameter tractable on chordal graphs parameterized by the number of colors; which answers an open question of Sampaio~\cite{SampaioThesis}.
\begin{corollary}
    \textsc{$b$-Coloring} can be solved in $2^{\calO(k^2)}n$ time on chordal graphs with $n$ vertices.
\end{corollary}
\begin{proof}
    Let $(G, k)$ be an instance of \textsc{$b$-Coloring} such that $G$ is a chordal graph.
	If the maximum clique size in $G$ is more than $k$, then $G$ has no proper coloring,
	and therefore no $b$-coloring, with $k$ colors.
	We may assume that the maximum clique size in $G$ is at most $k$.
	This in turn implies that the treewidth of $G$ is at most $k$,
	since a clique tree of $G$ (which can be found in linear time~\cite{BlairPeyton1993})
	is in fact a tree decomposition of width at most $k$ of $G$.
	We can therefore apply \cref{prop:tw:numcol}.
	
	Note that even though the algorithm of~\cite{BlairPeyton1993} implies a linear dependence on the number of edges in the input graph, this can be avoided by the following observation.
	If an $n$-vertex graph has tree-width at most $w$,
	then it has at most $wn$ edges.
	Therefore, if the number of edges in $G$ is more than $kn$ then we can report that $(G, k)$ is a \no-instance;
	otherwise, the dependence on the number of edges is subsumed by the run time of the algorithm from \cref{prop:tw:numcol}.
\end{proof}


\section{Parameterized by Clique-Width}
In this section, we consider the $b$-coloring problem parameterized by the clique-width of the input graph.
We will work with decompositions of bounded \emph{module-width}, which is equivalent for our purposes, 
see Theorem~\ref{thm:cw:mw}.

The main contribution of this section is an algorithm that given a graph $G$ on $n$ vertices and one of its
rooted branch decompositions of module-width $\givenmw$, and an integer $k$, decides whether $G$
has a $b$-coloring with $k$ colors in time $n^{2^{\calO(\givenmw)}}$.
Before we proceed, we observe that \bcol is \Wone-hard in this parameterization, and that
the exponential dependence on $\givenmw$ of the degree of the polynomial 
in the runtime is probably difficult to avoid.
\begin{proposition}\label{prop:bcol:lower:bound}
	The \bcol problem on graphs on $n$ vertices 
	parameterized by their module-width $\givenmw$ is \Wone-hard and
	cannot be solved in time $n^{2^{o(\givenmw)}}$, unless \ETH fails.
	Moreover, the hardness holds even when a linear branch decomposition of width $\givenmw$ is provided.
\end{proposition}
\begin{proof}
	Fomin et al.~\cite{Fomin2018} showed that the \textsc{Graph Coloring} problem 
	which given a graph $G$ of module-width $w$ and an integer $k$ asks for a proper coloring of $G$ with $k$ colors
	cannot be solved in time $n^{2^{o(\givenmw)}}$ unless \ETH fails, 
	even when a linear branch decomposition of module-width $w$ is provided.
	Using \textsc{Graph Coloring} in this setting as a starting point of a reduction, 
	we can add a $k$-clique to the input graph. 
	The resulting graph has a $b$-coloring with $k$ colors 
	if and only if the original graph has a proper coloring with $k$ colors
	(take the vertices in the $k$-clique as the $b$-vertices).
	It is not difficult to see that the given branch decomposition can be extended
	to include the vertices of the added $k$-clique without increasing its module-width
	by too much.
	\Wone-hardness parameterized by $\givenmw$ can be observed using the same argument,
	even as a consequence of an earlier result~\cite{FominEtAl2010}.
\end{proof}

\subsection{Outline of the Algorithm}
Throughout the following, we are given a graph $G$ and one of its rooted branch decompositions $(T, \decf)$
of module-width $\givenmw = \modulew(T, \decf)$ and we want to find a $b$-coloring of $G$ with $k$ colors,
if it exists.
In particular, our algorithm will find a $b$-coloring $\calC$ together with a set of \emph{witness $b$-vertices},
containing precisely one $b$-vertex for each color class of $\calC$, if it exists.
This will be done via dynamic programming along $T$, and for each node $t \in V(T)$, 
the partial solutions associated with $t$ are partial $b$-colorings of $G_t$ (recall Definition~\ref{def:partbcol}).

To obtain an efficient algorithm, we require a compact representation of the partial $b$-colorings of each
subgraph $G_t$ associated with a node $t \in V(T)$.
To that end, we introduce the notion of a \emph{$t$-signature} of a partial $b$-coloring.
Two partial $b$-colorings with the same $t$-signature will be interchangeable for the sake of our algorithm,
therefore the number of table entries at each node $t$ will be bounded by the number of $t$-signatures.

Let $(\calC, B)$ be a partial $b$-coloring of $G_t$.
For $(\calC, B)$ to be extended to a $b$-coloring $(\calC', B')$ of the entire graph $G$, 
we have to ensure that two things happen for each color class $C \in \calC$:
\begin{enumerate}
	\item\label{enum:bcol:cond:A} The extension of $C$ in $\calC'$ is an independent set in $G$.
	\item\label{enum:bcol:cond:B} There is a witness $b$-vertex in $B'$ for the extension of $C$ in $\calC'$.
\end{enumerate}

The $t$-signature has to represent a partial $b$-coloring faithfully enough 
so that we can keep track of all the ways in which the above two conditions can be satisfied for each of its color classes
`in the future'.
At the same time, its definition has to enable us to significantly compress 
the information about partial $b$-colorings of $G_t$.
This happens in the following way.
We categorize color classes of partial $b$-colorings of $G_t$ according to \emph{$t$-types}.
If two color classes $C_1$, $C_2$ of a partial $b$-coloring $(\calC, B)$ have the same $t$-type,
then the above two conditions can be satisfied for $C_1$ and $C_2$ by extensions of $(\calC, B)$ in the exact same ways.
This allows us to forget about the `names' of the color classes in a partial $b$-coloring, 
but instead to only remember for each $t$-type how many color classes with that type there are.
This is precisely the information that is stored in a $t$-signature.

Now, if we can bound the number of $t$-types by some function of the module-width $\givenmw$,
say~$f(w)$, then the number of $t$-signatures is upper bounded by $k^{f(w)} \le n^{f(w)}$.
(There are at most $k$ colors, so in particular there are at most $k$ colors with a given $t$-type.)
This translates directly to an upper bound on the number of table entries in the dynamic programming algorithm,
which, up to some constants in the degree of the polynomial, bounds the runtime of the resulting algorithm.

Let us discuss the information that goes into the definition of a $t$-type.
Let $C$ be a color class in a partial $b$-coloring $(\calC, B)$ of $G_t$.
To keep track of which vertices from $\overline{V_t}$ can be added to $C$ without introducing a coloring conflict,
it suffices to store which equivalence classes of $\sim_t$ have vertices in $C$,\footnote{This is similar to the algorithm 
of Wanke for \textsc{Graph Coloring} on graphs of bounded NLC-width~\cite{Wan94}.}
since all vertices in a given equivalence class have the same neighbors in $\overline{V_t}$.
This way we can ensure that condition~\ref{enum:bcol:cond:A} is satisfied.

To verify if condition~\ref{enum:bcol:cond:B} is satisfied
we have to store some information about the partial $b$-vertices.
Naturally, we record whether or not $B$ contains a partial $b$-vertex of $C$,
but we need to store more information.
Suppose that $B$ contains the partial $b$-vertex $v$ of $C$.
In a straightforward approach, 
we would simply keep track of the color classes that already appear in the neighborhood of $v$.
This way we could easily decide at which point during the execution of the algorithm, a partial $b$-vertex turns into a $b$-vertex.
However, this results in prohibitively large table entries, 
since there are $2^{k-1}$ subsets of colors that we would have to consider,
which for our purpose is no better than $2^n$.

We overcome this issue with the following symmetry breaking trick:
We do \emph{not} record which color classes the partial $b$-vertex of $C$ already sees/still needs to see.
Instead, we record which equivalence classes $Q \in V_t/{\sim_t}$ contain a partial $b$-vertex $w$ of \emph{some other color class} such that $N(w) \cap C = \emptyset$.
%
Suppose that some equivalence class $Q \in V_t/{\sim_t}$ contains the partial $b$-vertex $w \in B$ 
of another color class $C' \neq C$, such that $w$ has no neighbor of color $C$ in $V_t$.
For $w$ to become a $b$-vertex of its color,
the color class $C$ must be extended with a neighbor of $w$ in the future, i.e.\ in $\overline{V_t}$.
The neighborhood of $w$ in $\overline{V_t}$ is precisely $N_G(Q) \cap \overline{V_t}$,
therefore we can concisely model this situation as
color class $C$ requiring to contain a vertex among the future neighbors of $Q$.
In this situation, we say that 
\begin{center}
	\emph{color class $C$ has demand to the future neighbors of $Q$}.
\end{center}

The $t$-type records for each equivalence class $Q$ of $\sim_t$, 
if a color class contains vertices of $Q$, or if it has demand to the future of $Q$,
or none of the two.
Note that if a color class both contains a vertex from $Q$ and has demand to the future of $Q$,
we already know that we can disregard the corresponding partial $b$-coloring:
In the corresponding color class, 
we cannot add any future neighbors of $Q$ without creating a coloring conflict,
and if we do not add a future neighbor of $Q$, then there is some color class whose partial $b$-vertex
will never become a $b$-vertex.

Now, if we have a partial $b$-coloring in which every color class has a partial $b$-vertex,
and all demands have been fulfilled, meaning that there is no color class that has demand to 
the future of some equivalence class of $\sim_t$, then we know that we actually have a $b$-coloring.
Moreover, the number of $t$-types is $2^{\calO(\givenmw)}$,
so the resulting algorithm runs in time $n^{2^{\calO(\givenmw)}}$ (see above).

\subsection{$t$-Types and $t$-Signatures}
In this section we introduce the basic concepts that we alluded to in the above description,
namely the notion of a \emph{$t$-type} and of a \emph{$t$-signature}, 
where $t$ is some node in the given branch decomposition.
A $t$-type is meant to capture the necessary information of a color class in a partial $b$-coloring of $G_t$.
However, we cannot give the definition of a $t$-type as a property of a vertex set alone:
a color class $C$ may have demand to the future of an equivalence class, which is because there is a partial $b$-vertex
of \emph{another} color $C' \neq C$ that has no neighbor of color $C$ yet.
Therefore, we first give the definition of a $t$-type abstractly,
i.e.\ absent of any partial $b$-coloring or color class,
and then define what it means for a color class to be of a certain $t$-type \emph{within a partial $b$-coloring}.

The $t$-type is a pair of a bit that is meant to tell us whether or not a coloring contains a partial $b$-vertex of that color,
and a map that tells us for each equivalence class, whether there is a vertex of the color 
in the equivalence class (via the value $\ccontains$), or if the color has demand to the future neighbors 
of the equivalence class (via the value $\cdemand$), 
or none of the two (via the value $\cnone$).
\begin{definition}[$t$-Type]
	Let $G$ be a graph with rooted branch decomposition $(T, \decf)$ and let $t \in V(T)$.
	A \emph{$t$-type} is a pair $(\cdesc, \cbvtx)$ of a map 
	$\cdesc \colon Q_t/{\sim_t} \to \{\cnone,\ccontains,\cdemand\}$ and
	a bit $\cbvtx \in \{0, 1\}$.
	We denote the set of all $t$-types by $\ctypes_t$.
\end{definition}

Before we proceed, we observe an upper bound on the number of $t$-types.
For the component $\cbvtx$, we clearly only have two choices,
and for each equivalence class $Q$ of~$\sim_t$, the entry $\cdesc(Q)$ takes one of three values.
\begin{observation}\label{obs:number:of:types}
	Let $(T, \decf)$ be a rooted branch decomposition of module-width $\givenmw = \modulew(T, \decf)$.
	For each $t \in V(T)$, $\card{\ctypes_t} = 2\cdot 3^{\card{V_t/{\sim_t}}} \le 2\cdot 3^{\givenmw}$.
\end{observation}

We now define what it means for a color class 
to be of a certain $t$-type within a partial $b$-coloring of $G_t$.
This is basically a formalization of the above discussion, 
but it holds one aspect that is of importance of the algorithm and the arguments to follow.
We discuss this after the following definition,
which is illustrated in Figure~\ref{fig:bcol:type}.
\begin{figure}
	\centering
	\includegraphics[height=.17\textheight]{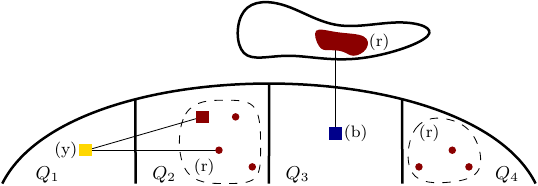}
	\caption{Illustration of the definition of a color class being of a certain \emph{$t$-type} inside a partial $b$-coloring of $G_t$.
			The large square vertices are partial $b$-vertices for their color.
			The type of the red (r) color in the coloring is as follows.
			Since it has a $b$-vertex (the one in $Q_2$), we have that $\cbvtx = 1$.
			Since $Q_2$ and $Q_4$ have red vertices, $\cdesc(Q_2) = \cdesc(Q_4) = \ccontains$.
			$Q_1$ and $Q_3$ do not have red vertices.
			$Q_1$ contains the $b$-vertex of color yellow (y),
			but this vertex already has a red neighbor. 
			Therefore, $\cdesc(Q_1) = \cnone$.
			Finally, $Q_3$ has the $b$-vertex of color blue (b),
			and this vertex does not have a red neighbor yet.
			Therefore, there has to be a red vertex among the future neighbors of $Q_3$.
			Hence, $\cdesc(Q_3) = \cdemand$.
			}
	\label{fig:bcol:type}
\end{figure}
\begin{definition}
	Let $G$ be a graph with rooted branch decomposition $(T, \decf)$ and let $t \in V(T)$.
	Let $(\calC, B)$ be a partial $b$-coloring of $G_t$,
	let $C \in \calC$ be a color class, and let $\ctype = (\cdesc, \cbvtx) \in \ctypes_t$ be a $t$-type.
	We say that \emph{$C$ has $t$-type $\ctype$ in $(\calC, B)$} if
	\begin{enumerate}
		\item $\cbvtx = \card{C \cap B}$ and
		\item\label{def:bcol:type2:2} 
		for each $Q \in V_t/{\sim_t}$,
		\begin{enumerate}
			\item if $Q \cap C \neq \emptyset$, and there is no $v \in (B \setminus C) \cap Q$ 
				such that $N(v) \cap C = \emptyset$,
				then $\cdesc(Q) = \ccontains$,
			\item if $Q \cap C = \emptyset$ and there exists some $v \in (B \setminus C) \cap Q$
				such that $N(v) \cap C = \emptyset$, 
				then $\cdesc(Q) = \cdemand$, and
			\item if $Q \cap C = \emptyset$, and there is no $v \in (B \setminus C) \cap Q$ 
				such that $N(v) \cap C = \emptyset$,
				then $\cdesc(Q) = \cnone$.
		\end{enumerate}
	\end{enumerate}
\end{definition}

The reader may have observed that~\ref{def:bcol:type2:2} does not cover all the possibilities.
The situation that is not covered is when $Q \cap C \neq \emptyset$ and there is some $v \in (B \setminus C) \cap Q$
such that $N(v) \cap C = \emptyset$.
A priori, we can of course not exclude this as a possibility, 
but there is a simple reason that partial $b$-colorings that contain a color class in which this situation arises 
can be disregarded:
For the vertex $v$ to become a $b$-vertex for its color, we have to add a future neighbor of $Q$ to $C$;
but since $Q$ already contains a vertex from $C$ this means that the resulting set is not independent anymore.

We turn to the definition of a $t$-signature
which again is first given in abstract terms.
\begin{definition}[$t$-Signature]\label{def:sig}
	Let $G$ be a graph with rooted branch decomposition $(T, \decf)$,
	and let $t \in V(T)$.
	A \emph{$t$-signature} is a map $\csig_t \colon \ctypes_t \to \{0, 1, \ldots, k\}$ 
	such that  $\sum_{\ctype \in \ctypes_t}\csig_t(\ctype) = k$.
\end{definition}

The following bound on the number of $t$-signatures immediately follows from
Observation~\ref{obs:number:of:types}: for each $t$-type, 
the function takes one of $k+1 \le n + 1$ values.
\begin{observation}\label{obs:number:of:sig}
	Let $G$ be a graph on $n$ vertices and $(T, \decf)$ be one of its branch decompositions 
	of module-width $\givenmw = \modulew(T, \decf)$.
	For each $t \in V(T)$, there are at most $n^{2^{\calO(\givenmw)}}$ many $t$-signatures.
\end{observation}

A $t$-signature \emph{represents} a partial $b$-coloring $(\calC, B)$ of $G_t$
if for each $t$-type it counts correctly how many color classes in $\calC$
are of that $t$-type in $(\calC, B)$.
\begin{definition}\label{def:sig:rep}
	Let $G$ be a graph with rooted branch decomposition $(T, \decf)$, 
	and let $t \in V(T)$.
	Let furthermore $\csig_t$ be a $t$-signature and $(\calC, B)$ a partial $b$-coloring in $G_t$.
	We say that \emph{$\csig_t$ represents $(\calC, B)$} 
	if for each $t$-type $\ctype \in \ctypes_t$,
	there are precisely $\csig_t(\ctype)$ color classes in $(\calC, B)$ that have $t$-type $\ctype$ in $(\calC, B)$.
	
	We call a partial $b$-coloring of $G_t$ \emph{representable} 
	if there is a $t$-signature that represents it.
\end{definition}

Since throughout this section, we only consider $b$-colorings and partial $b$-colorings with $k$ (possibly empty) colors, \cref{def:sig,def:sig:rep} together imply that if a partial $b$-coloring is represented by a $t$-signature, then necessarily each of its color classes has a $t$-type:
\cref{def:sig} requires that for a $t$-signature $\csig_t$, the sum of $\csig_t(\ctype)$ over all $t$-types $\ctype$ is $k$, and any partial $b$-coloring in $G_t$ has $k$ colors.

We would like to remark once more that not all partial $b$-colorings of $G_t$ can be represented by a $t$-signature, since there is a case that a color class cannot be described by a $t$-type.
In this case the partial $b$-coloring is not representable.
Conversely, we can make the following observation about representable partial $b$-colorings
which is useful in several proofs and sometimes used without explicit reference.
\begin{observation}\label{obs:bcol:representable:cont}
	Let $G$ be a graph with rooted branch decomposition $(T, \decf)$,
	and let $t \in V(T)$.
	Let $(\calC, B)$ be a representable partial $b$-coloring of $G_t$,
	and let $C \in \calC$ be a color class whose $t$-type in $(\calC, B)$ is $(\cdesc, \cbvtx)$.
	If for some equivalence class $Q \in V_t/{\sim_t}$, $Q \cap C \neq \emptyset$, then $\cdesc(Q) = \ccontains$.
\end{observation}

\subsection{Compatibility}
Let $t \in V(T)\setminus \leaves(T)$ be an internal node of the given rooted branch decomposition,
let $r$ and $s$ be its children, 
and let $(\decaux_t, \bubblemap_r, \bubblemap_s)$ be the operator of $t$.
In our algorithm, we want to combine information about partial $b$-colorings of
$G_r$ and $G_s$ to obtain information about partial $b$-colorings of $G_t$.
We will try to obtain a color class of a partial $b$-coloring 
of $G_t$ by taking the union of a color class $C_r$ of a partial $b$-coloring of $G_r$
and a color class $C_s$ of a partial $b$-coloring of $G_s$.

However, in some cases this is not possible. 
For instance, when $C_r$ contains vertices
from some equivalence class $Q_r \in V_r/{\sim_r}$ and $C_s$ contains vertices from some equivalence
class $Q_s \in V_s/{\sim_s}$, and in the graph $\decaux_t$ of the operator of $t$,
we have that $Q_rQ_s \in E(\decaux_t)$. 
Then, in $G_t$ all edges between the set $Q_r$ and $Q_s$ are present 
which means that $C_r \cup C_s$ is not an independent set in $G_t$.

Another condition is necessary to ensure that several demands that \emph{have to be} met
at node $t$ are indeed met.
Let $C_t = C_r \cup C_s$ and
suppose there is an equivalence class $Q_t \in V_t/{\sim_t}$
that contains a vertex of $C_t$.
Suppose furthermore that there is another equivalence class $Q_r \in V_r/{\sim_r}$ 
contained in the bubble of $Q_t$
such that $C_r$ has demand to the future neighbors of $Q_r$.
Then, this demand must be fulfilled by a neighbor of $Q_r$ in $C_s$
for otherwise, the equivalence class $Q_t$ 
both contains vertices of $C_t$ and $C_t$ has demand to the future neighbors of $Q_t$.
The resulting partial $b$-coloring would not be representable.

The following definition formalizes this discussion and projects it down to the `type level';
we illustrate this notion in Figure~\ref{fig:bcol:compatible:types}.
\begin{figure}
	\centering
	\includegraphics[height=.14\textheight]{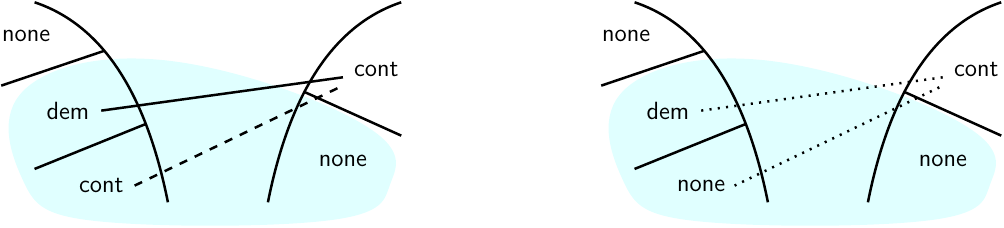}
	\caption{Illustration of Definition~\ref{def:compatibility}. 
		The shaded area shows a bubble and the labels on the equivalence classes correspond to type labelings.
		For the left hand side, note that between a pair of classes that are both labeled `$\ccontains$',
		there can be no edge in the operator.
		Moreover, since the bubble contains a class labeled $\ccontains$ and one labeled $\cdemand$,
		the demand of the latter has to be fulfilled at this node, i.e.\ there has to be an edge from this
		class to a `$\ccontains$'-class.
		The right side shows the situation when the `$\ccontains$'-class in the bubble is changed to `$\cnone$',
		in which case the dotted edges may or may not be present in the operator.
		}
	\label{fig:bcol:compatible:types}
\end{figure}
\begin{definition}[Compatible types]\label{def:compatibility}
	Let $G$ be a graph with rooted branch decomposition $(T, \decf)$.
	Let furthermore $t \in V(T) \setminus \leaves(T)$ with children $r$ and $s$, 
	and let $(\decaux_t, \bubblemap_r, \bubblemap_s)$ be the operator of $t$.
	Let $(\cdesc_r, \cbvtx_r) \in \ctypes_r$ and $(\cdesc_s, \cbvtx_s) \in \ctypes_s$.
	We say that $(\cdesc_r, \cbvtx_r)$ and $(\cdesc_s, \cbvtx_s)$ are 
	\emph{compatible} if the following conditions hold.
	\begin{enumerate}
		\item\label{def:compatibility:bvtx} $\cbvtx_r + \cbvtx_s \le 1$.
		\item\label{def:compatibility:proper} There is no pair $Q_r \in V_r/{\sim_r}$, $Q_s \in V_s/{\sim_s}$ 
			such that $Q_r Q_s \in E(\decaux_t)$ and
			$\cdesc_r(Q_r) = \cdesc_s(Q_s) = \ccontains$.
		\item\label{def:compatibility:demand}
			For each $Q \in V_t/{\sim_t}$ such that there exists a $p \in \{r, s\}$ and a $Q_p \in \bubblemap_p^{-1}(Q)$
			with $\cdesc_p(Q_p) = \ccontains$, the following holds.
		\begin{enumerate}
			\item For all $Q_r \in \bubblemap_r^{-1}(Q)$ with $\cdesc_r(Q_r) = \cdemand$,
				there is a $Q_s \in V_s/{\sim_s}$ with $\cdesc_s(Q_s) = \ccontains$ and $Q_rQ_s \in E(\decaux_t)$.
			\item Similarly, for all $Q_s \in \bubblemap_s^{-1}(Q)$ with $\cdesc_s(Q_s) = \cdemand$,
			there is a $Q_r \in V_r/{\sim_r}$ with $\cdesc_r(Q_r) = \ccontains$ and $Q_sQ_r \in E(\decaux_t)$.
		\end{enumerate}
	\end{enumerate}
\end{definition}

Given a pair of a color class $C_r$ of a partial $b$-coloring of $G_r$ and
a color class $C_s$ of a partial $b$-coloring of $G_s$ whose types in the respective colorings are compatible,
$C_r \cup C_s$, considered as a color class in a partial $b$-coloring of $G_t$,
has a fixed type.
We prove this later in the lemmas that attest the correctness of the algorithm,
but we already describe the construction of this type here,
mainly since the notion of compatibility of signatures that we give below,
requires this `merge type'.
\begin{definition}[Merge Type]\label{def:bcol:merge:type}
	Let $G$ be a graph with rooted branch decomposition $(T, \decf)$.
	Let furthermore $t \in V(T) \setminus \leaves(T)$ with children $r$ and $s$, 
	and let $(\decaux_t, \bubblemap_r, \bubblemap_s)$ be the operator of $t$.
	Let $\rho = (\cdesc_r, \cbvtx_r) \in \ctypes_r$ and $\sigma = (\cdesc_s, \cbvtx_s) \in \ctypes_s$
	be a pair of compatible types.
	The \emph{merge type} of $\rho$ and $\sigma$, denoted by $\mergetype(\rho, \sigma)$, 
	is the following $t$-type $(\cdesc_t, \cbvtx_t)$.
	\begin{enumerate}
		\item $\cbvtx_t = \cbvtx_r + \cbvtx_s$.
		\item For each $Q \in V_t/{\sim_t}$:
		\begin{enumerate}
			\item\label{bcol:merge:type:contains} 
				If for some $p \in \{r, s\}$, there exists a $Q_p \in \bubblemap_p^{-1}(Q)$ with $\cdesc_p(Q_p) = \ccontains$,
				then $\cdesc_t(Q) = \ccontains$.
			\item\label{bcol:merge:type:demand} 
				If \ref{bcol:merge:type:contains} does not apply and
				for some $p \in \{r, s\}$ there exists a $Q_p \in \bubblemap_p^{-1}(Q)$ with $\cdesc_p(Q_p) = \cdemand$
				and for $o \in \{r, s\} \setminus \{p\}$ and
				all $Q_pQ_o \in E(\decaux_t)$ we have $\cdesc_o(Q_o) \neq \ccontains$,
				then $\cdesc_t(Q) = \cdemand$.
			\item\label{bcol:merge:type:none}
				If neither~\ref{bcol:merge:type:contains} nor~\ref{bcol:merge:type:demand} applies,
				then $\cdesc_t(Q) = \cnone$.
		\end{enumerate}
	\end{enumerate}
\end{definition}

Towards a notion of compatibility of signatures, we first define a structure we call 
\emph{merge skeleton}. Given a node $t \in V(T)$ with children $r$ and $s$,
the merge skeleton is an edge-labeled bipartite graph whose vertices are the
$r$-types and the $s$-types, with the merge type of a compatible pair 
of types $\rho \in \ctypes_r$, $\sigma \in \ctypes_s$ written on the edge $\rho\sigma$.
Such an edge is meant to represent the fact that
taking the union of a color class $C_r$ that has $r$-type $\rho$ in a partial $b$-coloring of $G_r$
with a color class $C_s$ that has $s$-type $\sigma$ in a partial $b$-coloring of $G_s$
results in a color class of $t$-type $\mergetype(\rho, \sigma)$ in the partial $b$-coloring of $G_t$
that results from merging the partial $b$-colorings of $G_r$ and $G_s$.
\begin{definition}[Merge skeleton]
	Let $G$ be a graph and $(T, \decf)$ one of its rooted branch decompositions.
	Let $t \in V(T) \setminus \leaves(T)$ with children $r$ and $s$.
	The \emph{merge skeleton} of $r$ and $s$ is an edge-labeled bipartite graph
	$(\mergeaux, \malab)$ where
	\begin{itemize}
		\item $V(\mergeaux) = \ctypes_r \cup \ctypes_s$,
		\item for all $\rho \in \ctypes_r$, $\sigma \in \ctypes_s$, $\rho \sigma \in E(\mergeaux)$ 
			if and only if $\rho$ and $\sigma$ are compatible, and
		\item $\malab \colon E(\mergeaux) \to \ctypes_t$ is such that for all $\rho\sigma \in E(\mergeaux)$,
			$\malab(\rho\sigma)$ is the merge type of $\rho$ and $\sigma$.
	\end{itemize}
\end{definition}

Using the merge skeleton, we want to find out how to construct a $t$-signature of a partial $b$-coloring of $G_t$
that is obtained from a pair of a partial $b$-coloring for $G_r$ and one for $G_s$,
knowing only their signatures.
Any pair of an $r$-signature $\csig_r$ and an $s$-signature $\csig_s$ can `flesh out'
the merge skeleton $(\mergeaux, \malab)$ of $r$ and $s$, in the following sense.
We can obtain a map labeling the vertices of $\mergeaux$ that follows $\csig_r$ on $\ctypes_r$ and $\csig_s$ on $\ctypes_s$.
Then, an edge-labeling $\maassign$ of $\mergeaux$ with integers from $\{0, 1, \ldots, k\}$, 
such that for each vertex of $\mergeaux$,
the sum over its incident edges $e$ of $\maassign(e)$ is equal to its vertex label,
produces a $t$-signature $\csig_t$.
We can read off how many color classes of each type there are from the 
edge labeling $\maassign$.
In fact, each $t$-signature can be produced in such a way, as we prove below.
\begin{definition}[Compatible signatures]\label{def:compatible:sig}
	Let $(T, \decf)$ be a rooted branch decomposition.
	Let furthermore $t \in V(T) \setminus \leaves(T)$ with children $r$ and $s$.
	Let $\csignature_t$ be a $t$-signature,
	let $\csignature_r$ be an $r$-signature and
	$\csignature_s$ be a $s$-signature.
	We say that $(\csignature_t, \csignature_r, \csignature_s)$ is \emph{compatible} 
	if there is a triple $(\mergeaux, \malab, \maassign)$
	such that $(\mergeaux, \malab)$ is the merge skeleton of $r$ and $s$, and
	$\maassign \colon E(\mergeaux) \to \{0,1, \ldots, k\}$ is a map with the following properties.
	\begin{enumerate}
		\item\label{def:compatible:sig:children} For all $p \in \{r, s\}$ and all $\pi \in \ctypes_p$, 
			$\sum\nolimits_{e \in E(\mergeaux)\colon \pi \in e} \maassign(e) = \csignature_p(\pi)$.
		\item\label{def:compatible:sig:parent} For all $\ctype \in \ctypes_t$, 
			$\sum_{e \in E(\mergeaux)\colon \malab(e) = \ctype} \maassign(e) = \csignature_t(\ctype)$.
	\end{enumerate}
\end{definition}

We first show that we can test efficiently whether a triple of signatures is compatible.
\begin{lemma}\label{lem:bcol:runtime:compatibility}
	Let $G$ be a graph on $n$ vertices and
	let $(T, \decf)$ be one of its rooted branch decomposition of module-width $\givenmw = \modulew(T, \decf)$.
	Let $t \in V(T) \setminus \leaves(T)$ with children $r$ and $s$.
	Let $\csig_t$ be a $t$-signature, $\csig_r$ be an $r$-signature, and $\csig_s$ be an $s$-signature.
	One can decide in time $n^{2^{\calO(\givenmw)}}$ whether or not $(\csig_t, \csig_r, \csig_s)$ is compatible.
\end{lemma}
\begin{proof}
	We first observe that the merge skeleton can be constructed in $2^{\calO(\givenmw)}$ time,
	where $\givenmw = \modulew(T, \decf)$:
	It is easy to see that given two types $\rho \in \ctypes_r$, $\sigma \in \ctypes_s$,
	we can decide whether or not $\rho$ and $\sigma$ are compatible
	in time $\givenmw^{\calO(1)}$.
	Moreover, by Observation~\ref{obs:number:of:types}, $\card{\ctypes_r} \le 2^{\calO(\givenmw)}$
	and $\card{\ctypes_s} \le 2^{\calO(\givenmw)}$, therefore we have to check for 
	$(2^{\calO(\givenmw)})^2 = 2^{\calO(\givenmw)}$ pairs of types if they are compatible,
	and if so, compute their merge type.
	(This also implies that $\card{E(\mergeaux)} = 2^{\calO(\givenmw)}$.)
	Computing a merge type can be done in time $\givenmw^{\calO(1)}$ as well,
	simply by following the construction given in Definition~\ref{def:bcol:merge:type}.
	
	We brute-force all candidates for the labeling $\maassign$.
	Given such a candidate, we can verify in time $2^{\calO(\givenmw)}$
	if it satisfies parts~\ref{def:compatible:sig:children} and~\ref{def:compatible:sig:parent}
	of the definition of compatible signatures.
	Since $\card{E(\mergeaux)} = 2^{\calO(\givenmw)}$,
	a trivial upper bound on the number of such candidate labelings is $n^{2^{\calO(\givenmw)}}$
	and therefore the claimed bound follows.
\end{proof}

\subsection{Merging and Splitting Partial $b$-Colorings}\label{sec:bcol:merge:split}
In this section we show that the notions introduced above work as desired,
and the technical lemmas we prove here will be the cornerstone of the correctness proof of the resulting 
algorithm that we give later.

\subsubsection{Bottom to Top}

\begin{lemma}\label{lem:bcol:bt}
	Let $G$ be a graph with rooted branch decomposition $(T, \decf)$ and let 
	$t \in V(T) \setminus \leaves(T)$ be an internal node with children $r$ and $s$.
	Let $\csig_r$ be an $r$-signature, $\csig_s$ be an $s$-signature, and $\csig_t$ be a $t$-signature such that:
	\begin{itemize}
		\item For all $p \in \{r, s\}$, there is a partial $b$-coloring $(\calC_p, B_p)$ in $G_p$ that is represented by $\csig_p$, and
		\item $(\csig_t, \csig_r, \csig_s)$ is compatible.
	\end{itemize}
	Then, there is a partial $b$-coloring $(\calC_t, B_t)$ of $G_t$ that is represented by $\csig_t$.
\end{lemma}
\begin{proof}
	Let $(\mergeaux, \malab, \maassign)$ be the structure witnessing that $(\csig_t, \csig_r, \csig_s)$ is compatible.
	We use Algorithm~\ref{alg:merge:colorings} 
	to create the pair $(\calC_t, B_t)$. 
	We first show that $(\calC_t, B_t)$ is indeed a partial $b$-coloring of $G_t$,
	and then later that $\csig_t$ represents $(\calC_t, B_t)$.
	\begin{algorithm}
		\SetKwInOut{Input}{Input}
		\SetKwInOut{Output}{Output}
		\Input{$(\calC_r, B_r)$, $(\calC_s, B_s)$, $\mergeaux$, and $\maassign$ as above}
		\Output{$(\calC_t, B_t)$, where $\calC_t$ is a partition of $V_t$ and $B_t \subseteq V_t$.}
		$\calC_r' \gets \calC_r$, $\calC_s' \gets \calC_s$, $\calC_t \gets \emptyset$\;
		\ForEach{$\rho \in \ctypes_r$, $\sigma \in \ctypes_s$ with $\rho\sigma \in E(\mergeaux)$}{
			Let $x \gets \maassign(\rho\sigma)$\;
			\For{$i = 1, \ldots, x$}{
				\label{alg:merge:1} Let $C_r \in \calC_r'$ be of $r$-type $\rho$ and $C_s \in \calC_s'$ be of $s$-type $\sigma$\;
				\label{alg:merge:2} $\calC_t \gets \calC_t \cup \{C_r \cup C_s\}$\;
				\label{alg:merge:3} $\calC_r' \gets \calC_r' \setminus \{C_r\}$, $\calC_s' \gets \calC_s' \setminus \{C_s\}$\;
			}
		}
		\Return $(\calC_t, B_r \cup B_s)$\;
		\caption{Merging $(\calC_r, B_r)$ and $(\calC_s, B_s)$ according to $\mergeaux$ and $\maassign$.}
		\label{alg:merge:colorings}
	\end{algorithm}
	
	\begin{nestedclaim}\label{claim:bcol:bt:bcol}
		$(\calC_t, B_t)$ as constructed above is a partial $b$-coloring of $G_t$ with $k$ colors.
	\end{nestedclaim}
	\begin{claimproof}
	Since $\calC_r$ is a partition of $V_r$ and $\calC_s$ is a partition of $V_s$, 
	and each part of $\calC_r$ and $\calC_s$ is used precisely once to obtain a part of $\calC_t$ in Algorithm~\ref{alg:merge:colorings},
	it is clear by Definition~\ref{def:compatible:sig}(\ref{def:compatible:sig:children})
	that $\calC_t$ is a partition of $V_t$.
	Together with Definition~\ref{def:compatible:sig}(\ref{def:compatible:sig:parent})
	and the definition of a $t$-signature,
	this ensures that $\calC_t$ has $k$ parts.
	
	We argue that each part $C \in \calC_t$ is an independent set.
	Suppose for a contradiction that $C$ is not an independent set and let $uv \in E(G_t)$ be an edge with $u, v \in C$.
	By construction, there are $C_r \in \calC_r$ and $C_s \in \calC_s$ such that $C = C_r \cup C_s$.
	Moreover, since $C_r$ and $C_s$ are color classes in a coloring, they are independent sets, 
	so we may assume that $u \in C_r$ and $v \in C_s$ (up to renaming).
	For all $p \in \{r, s\}$,
	let $\ctype_p = (\cdesc_p, \cbvtx_p)$ be the $p$-type of $C_p$ in $(\calC_p, B_p)$.
	Let furthemore $Q_r \in V_r/{\sim_r}$ be the equivalence class of $\sim_r$ containing $u$
	and $Q_s \in V_s/{\sim_s}$ be the equivalence class of $\sim_s$ containing $v$.
	This means that $\cdesc_r(Q_r) = \cdesc_s(Q_s) = \ccontains$.
	For $u$ and $v$ to be adjacent, the edge $Q_rQ_s$ has to be present in $\decaux_t$.
	On the other hand, $\ctype_r \ctype_s$ is an edge of the merge skeleton
	which implies that $\ctype_r$ and $\ctype_s$ are compatible types;
	in which case Definition~\ref{def:compatibility}(\ref{def:compatibility:proper})
	forbids the presence of this edge in $\decaux_t$, a contradiction.
	
	We have shown that $\calC_t$ is a proper coloring of $G_t$,
	it remains to show that for all $C \in \calC_t$, $\card{C \cap B_t} \le 1$.
	Suppose for a contradiction that for some $C \in \calC_t$, $\card{C \cap B_t} > 1$,
	and let $C_r \in \calC_r$, $C_s \in \calC_s$ be such that $C = C_r \cup C_s$,
	as per Algorithm~\ref{alg:merge:colorings}.
	Since for all $p \in \{r, s\}$, $(\calC_p, B_p)$ is a partial $b$-coloring of $G_p$,
	we have that $\card{C_p \cap B_p} \le 1$, and clearly $C_r \cap B_s = C_s \cap B_r = \emptyset$.
	This means that $\card{C_r \cap B_r} = \card{C_s \cap B_s} = 1$;
	and in the $r$-type $(\cdesc_r, \cbvtx_r)$ of $C_r$ in $(\calC_r, B_r)$
	and the $s$-type $(\cdesc_s, \cbvtx_s)$ of $C_s$ in $(\calC_s, B_s)$,
	$\cbvtx_r = \cbvtx_s = 1$.
	But again, $(\cdesc_r, \cbvtx_r)$ and $(\cdesc_s, \cbvtx_s)$ are compatible,
	so by Definition~\ref{def:compatibility}(\ref{def:compatibility:bvtx}),
	$\cbvtx_r + \cbvtx_s \le 1$, a contradiction.
	\end{claimproof}
	
	To prove the lemma, it remains to show that the $t$-signature $\csig_t$ represents $(\calC_t, B_t)$.
	This is shown via the following claim, with Definition~\ref{def:compatible:sig}
	ensuring that the numbers work out.
	\begin{nestedclaim}
		Let $C_r \in \calC_r$ and $C_s \in \calC_s$,
		and let $\ctype_r = (\cdesc_r, \cbvtx_r)$ be the $r$-type of $C_r$ in $(\calC_r, B_r)$,
		and let $\ctype_s = (\cdesc_s, \cbvtx_s)$ be the $s$-type of $C_s$ in $(\calC_s, B_s)$,
		such that $C_t = C_r \cup C_s$ is a color class in $(\calC_t, B_t)$.
		Then, the $t$-type of $C_t$ in $(\calC_t, B_t)$ is $\mergetype(\ctype_r, \ctype_s)$.
	\end{nestedclaim}
	\begin{claimproof}
		First observe that if $C_t = C_r \cup C_s$ is a color class in $(\calC_t, B_t)$,
		then $\ctype_r$ and $\ctype_s$ are compatible by construction.
		Let $\ctype_t = (\cdesc_t, \cbvtx_t) = \mergetype(\ctype_r, \ctype_s)$.
		We have to argue that the $t$-type of $C_t$ in $(\calC_t, B_t)$ is indeed $(\cdesc_t, \cbvtx_t)$.
		
		For the first item of the definition of the merge type, we observe that $\xi_r + \xi_s = \card{C_r \cap B_r} + \card{C_s \cap B_s}$ and since $B_t = B_r \cup B_s$ and $C_t = C_r \cup C_s$, we have $\xi_t = \xi_r + \xi_s = \card{C_t \cap B_t}$.
		
		Now let $Q \in V_t/{\sim_t}$.
		Suppose that $\cdesc_t(Q) = \ccontains$; we have to argue that $C_t \cap Q \neq \emptyset$ 
		and that there is no vertex $v \in (B_t \setminus C_t) \cap Q$ with $N(v) \cap C_t = \emptyset$.
		By the definition of the merge type, there is some $p \in \{r, s\}$ such that 
		there is a $Q_p \in V_p/{\sim_p}$ with $\bubblemap_p(Q_p) = Q$ and $\cdesc_p(Q_p) = \ccontains$.
		Since $C_p$ has $p$-type $(\cdesc_p, \cbvtx_p)$ in $(\calC_p, B_p)$,
		$C_p \cap Q_p \neq \emptyset$ which implies that $C_t \cap Q \neq \emptyset$.
		Now suppose that there is some vertex $v \in (B_t \setminus C_t) \cap Q$ with $N(v) \cap C_t = \emptyset$.
		This means that there is some $p \in \{r, s\}$ and some $Q_p \in \bubblemap_p^{-1}(Q)$
		such that $v \in Q_p$, and $N(v) \cap C_p = \emptyset$.
		Since $C_p$ has a $p$-type in $(\calC_p, B_p)$,
		this means that $C_p \cap Q_p = \emptyset$ and therefore $\cdesc_p(Q_p) = \cdemand$.
		Assume wlog that $p = r$.
		Since $(\cdesc_r, \cbvtx_r)$ and $(\cdesc_s, \cbvtx_s)$ are compatible,
		we have by Definition~\ref{def:compatibility}(\ref{def:compatibility:demand})
		that there is some $Q_s \in V_s/{\sim_s}$ with $\cdesc_s(Q_s) = \ccontains$ and $Q_rQ_s \in E(\decaux_t)$.
		But this implies that $v$ has a neighbor in $C_s \subseteq C_t$,
		a contradiction.
		
		Now suppose that $\cdesc_t(Q) = \cdemand$. 
		By the definition of the merge type, we have that in this case:
		\begin{enumerate}
		    \item\label{claim:dem:1}
		    For any $p \in \{r, s\}$ and $Q_p \in V_p/{\sim_p}$ with $\bubblemap_p(Q_p) = Q$, $\cdesc_p(Q_p) \neq \ccontains$.
		    \item\label{claim:dem:2} 
		    We may assume (up to renaming) that for some $Q_r \in \bubblemap_r^{-1}(Q)$, $\cdesc_r(Q_r) = \cdemand$, 
		    \item\label{claim:dem:3}
		    and that for all $Q_rQ_s \in E(\decaux_t)$, $\cdesc_s(Q_s) \neq \ccontains$.
		\end{enumerate}
		From \cref{claim:dem:1} we derive that $C_t \cap Q = \emptyset$.
		Next, \cref{claim:dem:2} implies that there is a vertex $v \in (B_r \setminus C_r) \cap Q_r$ with $N(v) \cap C_r = \emptyset$,
		and by \cref{claim:dem:3}, we can conclude that $v$ has no neighbor in $C_s$ either. Therefore, $v$ has no neighbor in $C_t$, as required.
		
		Finally, suppose that $\cdesc_t(Q) = \cnone$.
		Again then there is no $Q_p \in \bubblemap^{-1}(Q)$ such that $\cdesc_p(Q_p) = \ccontains$.
		If for all $p \in \{r, s\}$ and all $Q_p \in \bubblemap_p^{-1}(Q)$, $\cdesc_p(Q_p) = \cnone$,
		then it is clear that $C_t \cap Q = \emptyset$, and that there is no $v \in (B_t \setminus C_t) \cap Q$
		with $N(v) \cap C = \emptyset$.
		So suppose (up to renaming) that for some $Q_r \in \bubblemap_r^{-1}(Q)$, $\cdesc_r(Q_r) = \cdemand$,
		implying that there is a vertex $v \in (B_r \setminus C_r) \cap Q_r$ with $N(v) \cap C_r = \emptyset$.
		Since we did not land in case~\ref{bcol:merge:type:demand} of the definition of a merge type,
		there is some $Q_r Q_s \in E(\decaux_t)$ such that $\cdesc_s(Q_s) = \ccontains$,
		which means $v$ has a neighbor in $C_s \subseteq C_t$.
		Since this holds for any such $Q_r$ (and $Q_s$), we can conclude that there is no vertex in $(B_t \setminus C_t) \cap Q$
		with $N(v) \cap C_t = \emptyset$. 
		This concludes the proof.
	\end{claimproof}
	
	This concludes the proof of Lemma~\ref{lem:bcol:bt}.
\end{proof}

\subsubsection{Top to Bottom}
\begin{lemma}\label{lem:bcol:tb}
	Let $G$ be a graph with rooted branch decomposition $(T, \decf)$ and let 
	$t \in V(T) \setminus \leaves(T)$ be an internal node with children $r$ and $s$.
	Let $\csig_t$ be a $t$-signature, 
	and suppose there is a partial $b$-coloring $(\calC_t, B_t)$ of $G_t$
	which is represented by $\csig_t$.
	Then, there exists an $r$-signature $\csig_r$ and an $s$-signature $\csig_s$ such that
	\begin{itemize}
		\item for all $p \in \{r, s\}$ there is a partial $b$-coloring $(\calC_p, B_p)$ represented by~$\csig_p$, and
		\item $(\csig_t, \csig_r, \csig_s)$ is compatible.
	\end{itemize}
\end{lemma}
\begin{proof}
	For all $p \in \{r, s\}$, we let $\calC_p \defeq \calC_t|_{V_p}$ and $B_p \defeq B_t \cap V_p$.
	It is clear that $(\calC_p, B_p)$ is a partial $b$-coloring of $G_p$.
	\begin{nestedclaim}\label{claim:bcol:tb:representable}
		For all $p \in \{r, s\}$, $(\calC_p, B_p)$ is represented by some $p$-signature.
	\end{nestedclaim}
	\begin{claimproof}
		Suppose the claim is false for $p = r$.
		Then there is some $C_r \in \calC_r$ that has no $r$-type in $(\calC_r, B_r)$, 
		meaning that for some $Q_r \in V_r/{\sim_r}$, $Q_r \cap C_r \neq \emptyset$
		and there is a vertex $v \in (B_r \setminus C_r) \cap Q_r$ with $N(v) \cap C_r = \emptyset$.
		By construction, there is a $C_t \in \calC_t$ with $C_t = C_r \cup C_s$ for some $C_s \subseteq V_s$.
		Since $(\calC_t, B_t)$ is representable, and $\bubblemap_r(Q_r) \cap C_t \neq \emptyset$, 
		we know that $N(v) \cap C_t \neq \emptyset$ (otherwise, $C_t$ has no $t$-type in $(\calC_t, B_t)$).
		Therefore, $N(v) \cap C_s \neq \emptyset$.
		But since all vertices in $Q_r$ are twins with respect to $V_s$,
		and since $C_r \cap Q_r \neq \emptyset$ and $v \in Q_r$, 
		this means that there is an edge between some vertex in $C_r$ and some vertex in $C_s$,
		contradicting the fact that $C_t$ is an independent set.
	\end{claimproof}
	
	By the previous claim, we know that $(\calC_r, B_r)$ is represented by some $r$-signature $\csig_r$
	and that $(\calC_s, B_s)$ is represented by some $s$-signature $\csig_s$.
	It remains to show that $(\csig_t, \csig_r, \csig_s)$ is compatible.
	To be able to argue this, we show that each $t$-type with non-zero value in $\csig_t$ appears as an edge label of the merge skeleton
	$(\mergeaux, \malab)$ of $r$ and $s$,
	in particular that it is the merge type of the $r$-type and $s$-type labeling the endpoints of this edge.
	\begin{nestedclaim}
		Let $C_t \in \calC_t$ be a color class whose $t$-type in $(\calC_t, B_t)$ is $\ctype_t = (\cdesc_t, \cbvtx_t)$.
		Let $\ctype_r = (\cdesc_r, \cbvtx_r)$ be the $r$-type of $C_r \defeq C_t \cap V_r$ in $(\calC_r, B_r)$,
		and let $\ctype_s = (\cdesc_s, \cbvtx_s)$ be the $s$-type of $C_s \defeq C_t \cap V_s$ in $(\calC_s, B_s)$.
		Then, $\ctype_r$ and $\ctype_s$ are compatible and $\ctype_t = \mergetype(\ctype_r, \ctype_s)$.
	\end{nestedclaim}
	\begin{claimproof}
		We first show that $\ctype_r$ and $\ctype_s$ are compatible.
		We know that $\cbvtx_t \in \{0, 1\}$ and that $\cbvtx_t = 1$ if and only if $\card{C_t \cap B_t} = 1$
		if and only if either $\card{C_r \cap B_r} = 1$ or $\card{C_s \cap B_s} = 1$
		if and only if either $\cbvtx_r = 1$ or $\cbvtx_s = 1$,
		therefore $\cbvtx_r + \cbvtx_s \le 1$, meaning that condition~\ref{def:compatibility:bvtx} 
		of the definition of compatibility is satisfied.
		Since $C_t$ is an independent set, there are no edges between $C_r$ and $C_s$.
		This means that for any pair $Q_r \in V_r/{\sim_r}$, $Q_s \in V_s/{\sim_s}$ with 
		$\cdesc_r(Q_r) = \cdesc_s(Q_s) = \ccontains$, $Q_r Q_s \notin E(\decaux_t)$,
		otherwise there would be an edge between $C_r$ and $C_s$, 
		so condition~\ref{def:compatibility:proper} is satisfied as well.
		
		Now suppose that Definition~\ref{def:compatibility}(\ref{def:compatibility:demand}) is violated.
		We may assume (up to renaming) that there is some $Q \in V_t/{\sim_t}$ with the following properties.
		There is a $Q_r^* \in \bubblemap_r^{-1}(Q)$ with $\cdesc_r(Q_r^*) = \ccontains$,
		meaning that $C_r \cap Q_r^* \neq \emptyset$ and so $C_t \cap Q \neq \emptyset$.
		Moreover, there is some $Q_r \in \bubblemap_r^{-1}(Q)$ with $\cdesc_r(Q_r) = \cdemand$,
		where for any $Q_r Q_s \in E(\decaux_t)$, $\cdesc_s(Q_s) \neq \ccontains$.
		This means that there is a vertex $v \in (B_r \setminus C_r) \cap Q_r$ such that $N(v) \cap C_r = \emptyset$,
		and moreover that $N(v) \cap C_s = \emptyset$, implying that $N(v) \cap C_t = \emptyset$.
		Note that $v \in (B_t \setminus C_t) \cap Q_t$.
		In other words, we have argued that $Q$ is an equivalence class of $\sim_t$
		such that $C_t \cap Q \neq \emptyset$ and there is a vertex $v \in (B_t \setminus C_t) \cap Q$ such that $N(v) \cap C_t = \emptyset$.
		But this means that the color class $C_t$ cannot have a $t$-type in $(\calC_t, B_t)$, so $(\calC_t, B_t)$ was not representable, a contradiction.
		
		Now we argue that $\ctype_t$, the $t$-type of $C_t$ in $(\calC_t, B_t)$, is indeed the merge type of $\ctype_r$ and $\ctype_s$.
		We already argued above that $\cbvtx_t = \cbvtx_r + \cbvtx_s$.
		Now let $Q \in V_t/{\sim_t}$, and suppose that $\cdesc_t(Q) = \ccontains$.
		This means that $Q \cap C_t \neq \emptyset$.
		We may assume (up to renaming) that $u \in Q_r \cap C_r$ for some $Q_r \in \bubblemap_r^{-1}(Q)$.
		Since $(\calC_r, B_r)$ is representable by Claim~\ref{claim:bcol:tb:representable} 
		this already implies that $\cdesc_r(Q_r) = \ccontains$,
		therefore $\cdesc_t(Q)$ is set in accordance with the definition of the merge type.
		
		Now suppose that for some $Q \in V_t/{\sim_t}$, $\cdesc_t(Q) = \cdemand$.
		Then, $Q \cap C_t = \emptyset$ and there is some $v \in (B_t \setminus C_t) \cap Q$
		such that $N(v) \cap C_t = \emptyset$.
		First, since $Q \cap C_t = \emptyset$, this immediately implies that for all 
		$p \in \{r, s\}$ and all $Q_p \in \bubblemap_p^{-1}(Q)$, $Q_p \cap C_p = \emptyset$
		and therefore $\cdesc_p(Q_p) \neq \ccontains$.
		Now for $p \in \{r, s\}$, let $Q_p$ be the equivalence class of $\sim_p$ containing $v$.
		We may assume (up to renaming) that $p = r$.
		Clearly, $\bubblemap_r(Q_r) = Q$, therefore $Q_r \cap C_r = \emptyset$.
		Moreover, $N(v) \cap C_r = \emptyset$, and we have that $\cdesc_r(Q_r) = \cdemand$.
		Now suppose for a contradiction that for some $Q_r Q_s \in E(\decaux_t)$,
		$\cdesc_s(Q_s) = \ccontains$.
		This implies that $N(v) \cap C_s \neq \emptyset$, and therefore $N(v) \cap C_t \neq \emptyset$, a contradiction.
		We have shown that also in this case, $\cdesc_t(Q)$ is set in accordance with the definition of the merge type.
		
		Finally, suppose that $\cdesc_t(Q) = \cnone$.
		Then, $Q \cap C_t = \emptyset$ and there is no $v \in (B_t \setminus C_t) \cap Q$ with $N(v) \cap C_t = \emptyset$.
		This immediately implies that for all $p \in \{r, s\}$ and all $Q_p \in \bubblemap^{-1}(Q)$,
		$\cdesc_p(Q_p) \neq \ccontains$.
		Suppose that for some $p \in \{r, s\}$ and some $Q_p \in \bubblemap_p^{-1}(Q)$, $\cdesc_p(Q_p) = \cdemand$,
		and assume (up to renaming) that $p = r$.
		This means that there is some vertex $u \in (B_r \setminus C_r) \cap Q_r$ with $N(u) \cap C_r = \emptyset$.
		On the other hand, we know that $N(u) \cap C_t \neq \emptyset$, so $u$ has a neighbor in $C_s$.
		This means that there is a $Q_rQ_s \in E(\decaux_t)$ such that $Q_s \cap C_s \neq \emptyset$,
		meaning that $\cdesc_s(Q_s) = \ccontains$.
		Therefore, $\cdesc_t(Q)$ is also set in accordance with the definition of the merge type.
	\end{claimproof}
	
	To finish the proof, we have to construct an edge labeling $\maassign \colon E(\mergeaux) \to \{0, 1, \ldots, k\}$
	satisfying the conditions of Definition~\ref{def:compatible:sig}.
	The previous claim tells us that we can construct $\maassign$ in a straightforward way.
	Initially, set $\maassign(\ctype_t) = 0$ for all $\ctype_t \in \ctypes_t$.
	For each color class $C_t \in \calC_t$ whose $t$-type in $(\calC_t, B_t)$ is $\ctype_t$,
	we know that the $r$-type of $C_r \defeq C_t \cap V_r$, say $\ctype_r$,
	and the $s$-type of $C_s \defeq C_t \cap V_s$, say $\ctype_s$, are such
	that $\ctype_t = \malab(\ctype_r\ctype_s)$, 
	i.e.\ $\ctype_t$ appears as the label of the edge between $\ctype_r$ and $\ctype_s$ in $\mergeaux$.
	We therefore increase the value of $\maassign(\ctype_r \ctype_s)$ by one.
	Once we did this for all color classes of $(\calC_t, B_t)$,
	the tuple $(\mergeaux, \malab, \maassign)$ satisfies the requirements of Definition~\ref{def:compatible:sig},
	so $(\csig_t, \csig_r, \csig_s)$ is compatible.
\end{proof}

\subsection{The Algorithm}
As alluded to above, the algorithm is bottom-up dynamic programming along 
the given rooted branch decomposition $(T, \decf)$ of $G$.
First, we define the table entries stored at each node.
\begin{dptabledef}
For a node $t \in V(T)$ and a $t$-signature $\csignature_t$, 
we let $\dptable[t, \csignature_t] = 1$ if and only if 
there exists a partial $b$-coloring of $G_t$ that is represented by $\csignature_t$.
\end{dptabledef}

We now show that if all table entries have been computed correctly, 
then the solution can be read off the table entries stored at the root $\rootnode$ 
of the given rooted branch decomposition. 
Observe that since $V_\rootnode = V(G)$ and therefore $\overline{V_\rootnode} = \emptyset$,
the equivalence relation $\sim_\rootnode$ has one equivalence class, namely $V(G)$.
\begin{lemma}\label{lem:bcol:root}
	Let $G$ be a graph with rooted branch decomposition $(T, \decf)$ 
	and let $\rootnode \in V(T)$ be the root of $T$.
	Let $\rho$ be the $\rootnode$-type $(\cdesc_\rootnode, \cbvtx_\rootnode)$
	with $\cbvtx_\rootnode = 1$ and $\cdesc_\rootnode(V(G)) = \ccontains$.
	Let $\csig_\rootnode$ be the $\rootnode$-signature letting $\csig_\rootnode(\rho) = k$.
	Then, $G$ has a $b$-coloring with $k$ colors if and only if $\dptable[\rootnode, \csignature_\rootnode] = 1$.
\end{lemma}
\begin{proof}
	Suppose that $G$ has a $b$-coloring $(\calC, B)$ with $k$ colors.
	Then, $(\calC, B)$ is also a partial $b$-coloring; but since all vertices in $B$ are already $b$-vertices for their color,
	all demands have been fulfilled. 
	This means that $(\calC, B)$ is representable by an $\rootnode$-signature,
	denote this $\rootnode$-signature by $\csig$.
	We argue that $\csig = \csig_\rootnode$,
	in particular that all color classes $C \in \calC$
	are of type $\rho = (\cdesc_\rootnode, \cbvtx_\rootnode)$ in $(\calC, B)$
	as in the statement of the lemma.
	Let $C \in \calC$ be any color class. 
	Since $(\calC, B)$ is a $b$-coloring, $B$ contains a $b$-vertex $v$ of $C$,
	therefore also $C \neq \emptyset$ which implies that the $\rootnode$-type of $C$ is indeed $\rho$.
	As this reasoning applies to all $k$ color classes of $(\calC, B)$,
	we can conclude that $\dptable[\rootnode, \csignature_\rootnode] = 1$.
	
	Now suppose for the other direction that $\dptable[\rootnode, \csig_\rootnode] = 1$.
	Then there is a partial $b$-coloring $(\calC, B)$ of $G_\rootnode = G$ with $k$ colors
	represented by $\csig_\rootnode$. 
	Since $(\cdesc_\rootnode, \cbvtx_\rootnode)$ is the type of each color class and $\cbvtx_\rootnode = 1$,
	each color class has a partial $b$-vertex;
	since no color class has demand to the future neighbors of $V(G)$ by $\cdesc_\rootnode$,
	each partial $b$-vertex is indeed a $b$-vertex for its color.
	Therefore, $\calC$ is a $b$-coloring of $G$ with $k$ colors.
\end{proof}

We describe how to compute the table entries, starting with the leaves of $T$.
\begin{dpleaves}
Let $t \in V(T)$ be a leaf node of $T$ and let $v \in V(G)$ be the vertex such that $\decf(v) = t$.
We show how to set the table entries $\dptable[t, \cdot]$.
The partial $b$-colorings of $G_t = (\{v\}, \emptyset)$ we have to consider are the following.
The vertex $v$ is colored with one of the $k$ colors, and it is either the partial $b$-vertex for its color or not.

The $t$-signatures representing these colorings look as follows.
Observe that $\sim_t$ has precisely one equivalence class, namely $\{v\}$.
We let $\cdesc_{\ccontains}$ be the map with $\cdesc_{\ccontains}(\{v\}) = \ccontains$.
In the case that $v$ is \emph{not} the partial $b$-vertex of its color, we have
\begin{itemize}
	\item one color of type $(\cdesc_{\ccontains}, 0)$, and
	\item $k - 1$ colors of type $(\cdesc_\emptyset, 0)$ with $\cdesc_\emptyset(\{v\}) = \cnone$.
\end{itemize}
We denote this signature by $\csignature_1$, i.e.\ we let
$\csignature_1((\cdesc_{\ccontains}, 0)) = 1$ and $\csignature_1((\cdesc_\emptyset, 0)) = k - 1$.

In the case that $v$ \emph{is} the partial $b$-vertex of its color class, 
then the remaining $k-1$ color classes have demand to the future neighbors of $\{v\}$,
so that $v$ eventually becomes the $b$-vertex of its color.
Therefore we have
\begin{itemize}
	\item one color of type $(\cdesc_{\ccontains}, 1)$, and
	\item $k - 1$ colors of type $(\cdesc_{\cdemand}, 0)$ with $\cdesc_{\cdemand}(\{v\}) = \cdemand$.
\end{itemize}
We denote this signature by $\csignature_2$, i.e.\ we let
$\csignature_2((\cdesc_{\ccontains}, 1)) = 1$ and $\csignature_2((\cdesc_{\cdemand}, 0)) = k - 1$.
To summarize, for each $t$-signature $\csignature$, we let
\begin{align*}
	\dptable[t, \csignature] \defeq \left\lbrace
		\begin{array}{ll}
			1, &\mbox{if } \csignature \in \{\csignature_1, \csignature_2\} \\
			0, &\mbox{otherwise}
		\end{array}
		\right.
\end{align*}
\end{dpleaves}

Next, the internal nodes of $T$.
\begin{dpinternal}
Now let $t \in V(T) \setminus \leaves(T)$ with children $r$ and $s$.
For each $t$-signature $\csignature_t$, we let $\dptable[t, \csignature_t] = 1$ if and only if
there exists a pair $(\csignature_r, \csignature_s)$ 
of an $r$-signature $\csignature_r$ and an $s$-signature $\csignature_s$
such that
\begin{enumerate}
	\item\label{enum:bcol:join:children} $\dptable[r, \csignature_r] = 1$ and $\dptable[s, \csignature_s] = 1$, and
	\item\label{enum:bcol:join:compatible} $(\csignature_t, \csignature_r, \csignature_s)$ is compatible. 
\end{enumerate}
\end{dpinternal}

Equipped with the lemmas of the previous sections, we can prove correctness of the above algorithm.
\begin{lemma}\label{lem:bcol:correctness}
	For each $t \in V(T)$ and $t$-signature $\csignature_t$, 
	the above algorithm computes the table entry $\dptable[t, \csignature_t]$ correctly.
\end{lemma}
\begin{proof}
	We prove the lemma by induction on the height of $t$.
	For the base case, when $t$ is a leaf, it is easily verified.
	From now on we may assume that $t \in V(T) \setminus \leaves(T)$ with children $r$ and $s$.
	
	First, suppose that the algorithm set $\dptable[t, \csignature] = 1$.
	This means that there is a pair $(\csig_r, \csig_s)$ of an $r$-signature $\csig_r$ and an $s$-signature $\csig_s$
	such that $\dptable[r, \csig_r] = 1$ and $\dptable[s, \csig_s] = 1$ and $(\csig_t, \csig_r, \csig_s)$ is compatible.
	By induction, we know that there is a partial $b$-coloring of $G_r$ represented by the $r$-signature $\csig_r$
	and a partial $b$-coloring of $G_s$ represented by the $s$-signature $\csig_s$.
	Then, by Lemma~\ref{lem:bcol:bt}, there is a partial $b$-coloring of $G_t$ represented by the $t$-signature $\csig_t$.
	
	Conversely, suppose that there is a partial $b$-coloring of $G_t$ represented by the $t$-signature $\csig_t$.
	Then, by Lemma~\ref{lem:bcol:tb}, there is a partial $b$-coloring of $G_r$ represented by an $r$-signature $\csig_r$
	and a partial $b$-coloring of $G_s$ represented by an $s$-signature $\csig_s$,
	such that $(\csig_t, \csig_r, \csig_s)$ is compatible.
	By induction, the algorithm set $\dptable[r, \csig_r] = 1$ and $\dptable[s, \csig_s] = 1$,
	and therefore, by the above description, it set $\dptable[t, \csig_t] = 1$.
\end{proof}

We wrap up. 
By Lemma~\ref{lem:bcol:correctness}, the algorithm computes all table entries correctly,
and by Lemma~\ref{lem:bcol:root}, the solution to the instance can be determined upon inspecting
the table entries associated with the root of the given branch decomposition.
Correctness of the algorithm follows.

Regarding the runtime, we observe the following.
Given an $n$-vertex graph with rooted branch decomposition $(T, \decf)$ of module-width $\givenmw = \modulew(T, \decf)$,
we have that $\card{V(T)} = \calO(n)$.
($T$ is a full binary tree on $n$ leaves, so $\card{V(T)} = 2n - 1$.)
Let $t \in V(T)$. 
If $t$ is a leaf node, then computing the table entries $\dptable[t, \cdot]$ takes constant time.
If $t$ is an internal node, then by Observation~\ref{obs:number:of:sig}, 
we have to compute $n^{2^{\calO(\givenmw)}}$ table entries.
Assume by induction that the table entries associated with the children of $t$ have been computed.
For each $t$-signature $\csig_t$
we have to try for $\left(n^{2^{\calO(\givenmw)}}\right)^2 = n^{2^{\calO(\givenmw)}}$ pairs of one signature per child
whether or not they form a compatible triple together with $\csig_t$.
For each triple, this can be done in time $n^{2^{\calO(\givenmw)}}$ by Lemma~\ref{lem:bcol:runtime:compatibility}.
Therefore, the overall runtime of the algorithm is $n^{2^{\calO(\givenmw)}}$.
\begin{theorem}\label{thm:alg:cw}
	There is an algorithm that solves \bcol in time $n^{2^{\calO(\givenmw)}}$, where
	$n$ denotes the number of vertices of the input graph, and
	$\givenmw$ denotes the module-width of a given rooted branch decomposition of the input graph.
\end{theorem}


\subsection{Fall Coloring}\label{sec:cw:fall:col}
Recall that a \emph{fall coloring} is a special type of $b$-coloring where 
\emph{every} vertex is a $b$-vertex for its color.
In other words, it is a partition of the vertex set of a graph into independent dominating sets.
We adapt our algorithm for \bcol on graphs of bounded clique-width to solve 
\fallcol, and therefore show that the latter problem is as well solvable in time 
$n^{2^{\calO(\givenmw)}}$, where $\givenmw$ denotes the clique-width of a given 
decomposition of the input graph.

\subsection*{Adaptation of the \bcol Algorithm}
We now show how to adapt the algorithm of Theorem~\ref{thm:alg:cw} to solve the \fallcol problem
in time $n^{2^{\calO(\givenmw)}}$ as well. 
This adaptation in some sense simplifies the algorithm for \bcol, since we do no have to keep track 
of whether or not a color class has a $b$-vertex in a partial coloring;
\emph{every} vertex has to be a $b$-vertex.
Now, if we can construct a coloring such that each color class is nonempty, 
and each vertex is a $b$-vertex for its color,
then clearly we have a fall coloring.
With small modification, the mechanics of our algorithm for \bcol 
allow for checking if there is a coloring with this property.
The main difference will be in the definition of the type of a color class.

Let $(C_1, \ldots, C_k)$ be a proper coloring of $G_t$ for some node $t$,
and $C_i$ and $C_j$ be two distinct color classes.
If for some $Q \in V_t/{\sim_t}$, $C_i \cap Q = \emptyset$, and there is \emph{any} vertex $v_j \in C_j$
such that $N(v_j) \cap C_i = \emptyset$,
then $C_i$ has demand to the future neighbors of $Q$: 
the vertex $v_j$ needs to become a $b$-vertex of color $j$,
and since it has no neighbor in color class $i$ so far, one of its future neighbors
(equivalently, a future neighbor of equivalence class $Q$), has to receive color $i$.

The definition of a \emph{$t$-fall type} can be obtained from the definition of a $t$-type
by dropping the bit $\cbvtx$ which becomes unnecessary in the context of \fallcol.

The definition of a color class being of a certain $t$-fall type becomes the following.
\begin{definition}[$t$-Fall-type]
	Let $G$ be a graph with rooted branch decomposition $(T, \decf)$, and let $t \in V(T)$.
	A \emph{$t$-fall type} is a map $\cdesc \colon V_t/{\sim_t} \to \{\cnone,\ccontains,\cdemand\}$.
	
	Let $\calC = (C_1, \ldots, C_k)$ be a proper coloring of $G_t$, and let $\cdesc$ be a $t$-fall type.
	For $i \in \{1, \ldots, k\}$, we say that \emph{$C_i$ has $t$-fall type $\cdesc$ in $\calC$}
	if for each $Q \in V_t/{\sim_t}$,
	\begin{enumerate}
		\item\label{def:fall:type:desc:contains} 
		if $Q \cap C_i \neq \emptyset$ and for all $v \in Q \setminus C_i$, $N(v) \cap C_i \neq \emptyset$, then $\cdesc(Q) = \ccontains$,
		\item\label{def:fall:type:desc:demand} 
		if $Q \cap C_i = \emptyset$ and 
			there is a $v \in Q \setminus C_i$ 
			with $N_{G_t}(v) \cap C_i = \emptyset$, 
			then $\cdesc(Q) = \cdemand$, and
		\item\label{def:fall:type:desc:none} $\cdesc(Q) = \cnone$, otherwise.
	\end{enumerate}
\end{definition}

We again restrict ourselves to finding (partial) colorings that are \emph{representable},
in the sense that there is no color class that both intersects an equivalence class and has demand to its future neighbors.
In complete analogy, we define a $t$-signature as a function counting the number of color classes of each $t$-fall type.

We say that two fall-types are compatible, if they satisfy parts~\ref{def:compatibility:proper} and~\ref{def:compatibility:demand} of Definition~\ref{def:compatibility},
the definition of compatible types in the case of \bcol.
Part~\ref{def:compatibility:bvtx} simply disappears since we do not have to keep track 
of whether or not a color class contains a partial $b$-vertex.
With this in mind, the technical arguments given in Section~\ref{sec:bcol:merge:split} go through.

The definition of the table entries is analogous as well, and
by an argument parallel to the proof of Lemma~\ref{lem:bcol:root},
we can conclude that this information is sufficient to solve the problem.

We discuss the resulting algorithm.
For the leaf nodes, we only have to consider colorings with one color class whose fall-type is $\cdesc_v(\{v\}) = \ccontains$
and $k-1$ color classes whose fall-type is $\cdesc_\cdemand(\{v\}) = \cdemand$.
This is because in any fall-coloring of $G$, the vertex $v$ has to be a $b$-vertex for its color.
The computation of the internal nodes remains the same.
A correctness proof of the algorithm can now be given in the same way as in the proof of Lemma~\ref{lem:bcol:correctness},
and the discussion of the runtime of the algorithm still goes through.
We have the following theorem.
\begin{theorem}\label{thm:alg:fall:cw}
	There is an algorithm that solves \fallcol in time $n^{2^{\calO(\givenmw)}}$, where
	$n$ denotes the number of vertices of the input graph, and
	$\givenmw$ denotes the module-width of a given rooted branch decomposition of the input graph.
\end{theorem}

\subsection*{Hardness}
We now show that the runtime of the algorithm from Theorem~\ref{thm:alg:fall:cw} is optimal in some sense.
Specifically, we give a reduction that proves the same lower bounds as the ones we obtained for \bcol.
Recall again that linear module-width and linear clique-width can be used interchangeably in this setting
(Theorem~\ref{thm:cw:mw}).
\begin{proposition}
	The \fallcol problem on graphs on $n$ vertices 
	parameterized by the module-width $\givenmw$ of the input graph is \Wone-hard and
	cannot be solved in time $n^{2^{o(\givenmw)}}$, unless \ETH fails.
	Moreover, the hardness holds even when a linear branch decomposition of width $\givenmw$ is provided.
\end{proposition}
\begin{proof}
	We give a reduction from \textsc{Graph Coloring} parameterized by the module-width $w$ of the input graph 
	which is \Wone-hard and has no $n^{2^{o(w)}}$-time algorithm under \ETH~\cite{FominEtAl2010,Fomin2018}.
	Given an instance $(G, k)$
	construct a instance $(H, k)$ of \textsc{Fall Coloring} as follows.
	We obtain $H$ from $G$ by adding, for each vertex $v \in V(G)$, a clique $X_v$ on $k-1$ vertices to the graph,
	and making $X_v$ complete to $v$.

	If $H$ has a fall coloring with $k$ colors, then clearly this is a proper coloring of $G$ with $k$ colors,
	since $G$ is an induced subgraph of $H$.
	Suppose $G$ has a proper coloring with $k$ colors.
	For each vertex $v \in V(G)$, we can bijectively assign the $k-1$ remaining colors 
	(i.e.\ all colors except the one appearing on $v$)
	to the vertices of $X_v$.
	The coloring constructed this way is a fall coloring of $H$ with $k$ colors:
	First, we immediately observe that the coloring is proper. 
	Since we started from a proper coloring of $G$, there is no monochromatic edge in $G$.
	Since we colored the vertices of each $X_v$ bijectively with all colors except the one appearing on $v$,
	and since $N_H(X_v) \cap V(G) = \{v\}$, we did not introduce any monochromatic edge either.
	It remains to argue that each vertex of $H$ is a $b$-vertex for its color.
	For each $v \in V(G)$, we have that $v$ is a $b$-vertex since the remaining $k-1$ colors appear on $X_v$.
	For each $u \in X_v$, we have that $u$ is a $b$-vertex since it sees $k-2$ colors on $X_v \setminus \{u\}$, 
	plus the color of $v$; since $X_v \cup \{v\}$ is a clique, all of these colors are mutually distinct.

	The size of $H$ is polynomial in the size of $G$, 
	and it is clear that adding the cliques $X_v$ did not increase the module-width of $G$.
\end{proof}


\section{Conclusion}
In this work, we gave an \XP-algorithm for \bcol parameterized by the clique-width 
of a given decomposition of the input graph, and an \FPT-algorithm parameterized
by the vertex cover number.
This initiated the study of structural parameterizations of the \bcol and \bchrom
problems.
The most prominent parameter sitting between clique-width and the vertex cover number is
arguably the treewidth of a graph.
Since any graph of bounded treewidth has bounded clique-width,
our algorithm implies that \bcol parameterized by treewidth is in \XP.
We therefore ask, is \bcol parameterized by the treewidth of the input graph \FPT or \W[1]-hard?

It would be interesting to obtain an \FPT-algorithm for \bcol parameterized by the vertex 
cover number $\vertexcover$ whose runtime is tight under \ETH.
Lokshtanov et al.~\cite{LMS2018b} showed that \gcol has no 
$2^{o(\vertexcover \log \vertexcover)}\cdot n^{\calO(1)}$ time algorithm unless \ETH fails,
and by the same argument\footnote{Noting that we may assume that 
the number of colors is always linearly bounded in the vertex cover number;
so adding the clique does not increase the number of colors in a prohibitive way.} 
given in Proposition~\ref{prop:bcol:lower:bound}, this rules out 
$2^{o(\vertexcover \log \vertexcover)}\cdot n^{\calO(1)}$ time algorithms for \bcol under \ETH.
We therefore ask if the runtime of
$2^{\calO(\vertexcover^2)}\cdot n^{\calO(1)}$ in
Corollary~\ref{coro:vc} can be improved to
$2^{\calO(\vertexcover \log \vertexcover)}\cdot n^{\calO(1)}$.

There are two main approaches for solving \textsc{Graph Coloring} parameterized by clique-width,
one being efficient when the number of colors is small~\cite{Lampis2020},
and the other being efficient when the number of colors is large~\cite{EspelageEtAl2001,Wan94}.
Our algorithm for \textsc{$b$-Coloring} falls in the latter category.
It would be interesting to obtain an efficient algorithm 
for \bcol parameterized by clique-width when the number is small,
with a running time that is tight under the Strong Exponential Time Hypothesis
as it was done for \textsc{Graph Coloring} by Lampis~\cite{Lampis2020}.
Moreover, Courcelle et al.~\cite{CourcelleEtAl2020} recently gave an algorithm that unifies both 
approaches into a single algorithm; is the same possible for \bcol?

\paragraph*{Acknowledgements.}
    We would like to thank the anonymous reviewers for many valuable suggestions that improved this work.
    We are particularly grateful for the suggestion to replace our initial vertex cover based algorithm and the use of Courcelle's Theorem by a single explicit DP algorithm on graphs of bounded tree-width. This led to a cleaner and more precise presentation of the result on chordal graphs as well as an improved algorithm parameterized by vertex cover.

\bibliographystyle{plain}
\bibliography{99-ref}   

\end{document}